\documentclass[11pt]{article}

\usepackage{balance} % for balancing columns on the final page
\usepackage[numbers]{natbib}
\usepackage{authblk}
\usepackage{multirow}
\usepackage{makecell}
\usepackage{tabu}
\usepackage{graphicx}
\usepackage{enumitem}
\usepackage{algorithm}
\usepackage{algpseudocode}
\usepackage{blindtext}
\usepackage{hyperref}
\usepackage{xcolor}
\usepackage{amsfonts}
\usepackage[english]{babel}
\usepackage{amsthm}
\usepackage{amsmath}
\usepackage{cleveref}
\usepackage{array}
\usepackage{booktabs}
\usepackage{soul}
\usepackage[bibliography=common]{apxproof}

\newcommand{\ww}{\mathbf{w}}
\newcommand{\gcdw}{\operatorname{gcd}(\ww)}

\usepackage{color}

\newcommand{\erel}[1]{}

\newcommand{\BibTeX}{\rm B\kern-.05em{\sc i\kern-.025em b}\kern-.08em\TeX}

\newtheoremrep{theorem}{Theorem}[section]
\newtheoremrep{proposition}[theorem]{Proposition}
\newtheoremrep{lemma}[theorem]{Lemma}
\newtheorem{observation}[theorem]{observation}

\theoremstyle{definition}
\newtheorem{example}[theorem]{Example}
\newtheorem{definition}[theorem]{Definition}
\newtheorem{remark}[theorem]{Remark}



\title{Weighted Envy Freeness With Bounded Subsidies \footnote{A short version of this article can be found at \cite{elmalem2024weighted}.}}

\author[1]{Noga Klein Elmalem}
\author[1] {Rica Gonen}
\author[2]{Erel Segal-Halevi}
\affil[1]{The Open University of Israel}
\affil[2]{Ariel University}

\begin{document}
\maketitle

\begin{abstract}
We explore solutions for fairly allocating indivisible items among agents assigned weights representing their entitlements.
Our fairness goal is \textbf{weighted-envy-freeness (WEF)}, where each agent deems their allocated portion relative to their entitlement at least as favorable as any other’s relative to their own.
In many cases, achieving WEF necessitates monetary transfers, which can be modeled as third-party subsidies. The goal is to attain WEF with bounded subsidies.

Previous work in the unweighted setting of subsidies relied on basic characterizations of EF that fail in the weighted settings. This makes our new setting challenging and theoretically intriguing. We present polynomial-time algorithms that compute WEF-able allocations with an upper bound on the subsidy per agent in three distinct additive valuation scenarios: (1) general, (2) identical, and (3) binary. %Our algorithms’ subsidy bounds are tight in the worst case.
When all weights are equal, our bounds reduce to the bounds derived in the literature for the unweighted setting.
\end{abstract}

%%%%%%%%%%%%%%%%%%%%%%%%%%%%%%%%%%%%%%%%%%%%%%%%%%%%%%%%%%%%%%%%%%%%%%%%

\section{Introduction}

The mathematical theory of fair item allocation among multiple agents has practical applications in scenarios like inheritance and partnership dissolutions. When agents have equal entitlements, as in inheritance cases, each agent naturally expects their allotment to be at least as good as others'. An allocation satisfying this requirement is called \emph{envy-free (EF)}. 

When the items available for allocation are indivisible, an EF allocation might not exist. A solution often applied in practice is to use \emph{money} to compensate for the envy. In the recent literature, it is common to assume that a hypothetical third-party is willing to subsidize the process such that all agents receive a non-negative amount, and ask what is the \emph{minimum amount of subsidy} required to attain envy-freeness. 

It is common to assume that the agents are \emph{quasilinear}. This means that their utility equals their total value for the items they receive, plus the amount of money they receive (which may be positive or negative). 

%Alternatively, one can consider a hypothetical third-party, who is willing to subsidize the process such that all agents receive a non-negative amount, and ask what is the \emph{minimum amount of subsidy} that is required to attain envy-freeness. 

The subsidy minimization problem was first studied by \citet{halpern2019fair}. They showed that, for any given allocation, there exists a permutation of the bundles that is \emph{envy-freeable (EF-able)}, that is, can be made \emph{envy-free} with subsidies. The total required subsidy is at most $(n-1)mV$, where $m$ is the number of items, $n$ the number of agents, and $V$ the maximum item value for an agent, and this bound is tight in the worst case. \citet{brustle2020one} considered the case in which the allocation is not given, but can be chosen. They presented an algorithm that finds an envy-freeable allocation through iterative maximum matching, requiring a total subsidy of at most $(n-1) V$, which is  tight too.

This paper extends previous work by considering agents with different entitlements, which we call \emph{weights}. This extension is useful in partnership dissolutions, where partners often hold varying numbers of shares, entitling them to different proportions of the asset. In such cases, each agent expects to receive at least the same "value per share" as others. 
For example, if agent $i$ has twice the entitlement of agent $j$, $i$ expects a bundle worth at least twice as much as $j$'s.

Formally, an allocation is called \emph{weighted envy-free (WEF)} (see e.g. \citet{robertson1998cake,zeng2000approximate,CISZ21}) if for every two agents $i$ and $j$, $\frac{1}{w_i}$ times the utility that $i$ assigns to his own bundle is at least as high as $\frac{1}{w_j}$ times the utility that $i$ assigns to the bundle of $j$, where $w_i$ is $i$'s entitlement and $w_j$ is $j$'s entitlement.

Now, we define \emph{weighted envy-freeability (WEF-ablity)}, the key concept we propose, analogously to the unweighted case: an allocation is WEF-able if it can be made WEF with subsides. More presicely, an allocation is WEF-able if for every two agents $i$ and $j$, $\frac{1}{w_i}$ times the sum of the utility that $i$ assigns to his own bundle and the subsidy he receives is at least as high as $\frac{1}{w_j}$ times the sum of the utility that $i$ assigns to the bundle of $j$ and the subsidy $j$ receives, where $w_i$ (resp., $w_j$) is $i$'s (resp., $j$'s) entitlement. 

Here, we assume quasi-linear utilities.
To illustrate the difficulty in this generalized setting, we show that the results from \citet{halpern2019fair,brustle2020one} fail when agents have different entitlements.
\begin{example} \label{example_intro}
There are two items $o_1,o_2$ and two agents $i_{1}, i_{2}$, with weights $w_{1} = 1, w_{2} = 10$ 
and valuation functions 
$$
\begin{bmatrix}
   & o_1 & o_2 \\
  i_1 & 5 & 7 \\
  i_2 & 10 & 8 \\
\end{bmatrix}
$$
We will show that, contrary to the result of \citet{halpern2019fair}, there exists a division of items where no permutation satisfies WEF.
Consider the bundles $A_1 = \{o_1\}$ and $A_2 = \{o_2\}$, where $i_1$ receives $A_1$ and $i_2$ receives $A_2$. Let $s_1$ and $s_2$ represent the subsidies for $i_1$ and $i_2$, respectively. The utility of $i_1$ for their own bundle is $5 + s_1$, and for $i_2$'s bundle, it is $7 + s_2$. To satisfy WEF, we need:
$\frac{5+s_1}{1} \geq \frac{7+s_2}{10}$, which implies $s_2 \leq 43 + 10 s_1$. 
Similarly, for agent $i_2$, WEF requires: $\frac{8+s_2}{10} \geq \frac{10+s_1}{1}$, which implies $s_2 \geq 92 + 10 s_1$. These two conditions are contradictory, so no subsidies can make this allocation WEF.
Next, consider the permutation where $i_1$ receives $A_2$ and $i_2$ receives $A_1$. In this case, WEF requires: $\frac{7 + s_1}{1} \geq \frac{5+s_2}{10}$, which implies $s_2 \leq 65 + 10 s_1$, 
and for $i_2$: $\frac{10 +s_2}{10} \geq \frac{8 + s_1}{1}$, which implies $s_2 \geq 70 + 10 s_1$. Again, these conditions are contradictory, proving that no permutation of bundles satisfies WEF.

This example also shows that the Iterated Maximum Matching algorithm of \citet{brustle2020one} does not guarantee WEF. The algorithm yields an allocation where all agents receive the same number of items, but as shown, no such allocation can be made WEF.
\end{example}

Of course, since the unweighted case is equivalent to the weighted case where each weight $w_i = 1/n$, all
negative results from the unweighted setting extend to the weighted case. In particular, it is NP-hard to compute the minimum subsidy required to achieve (weighted) envy-freeness, even in the binary additive case (as shown in \cite[Corollary 1]{halpern2019fair}). 
Thus, following previous work,
we develop polynomial-time algorithms that, while not necessarily optimal, guarantee an upper bound on the total subsidy. 
% While our allocations may not be optimal, they represent the best-known solutions available.
\subsection{Related Work}
\paragraph{\textbf{Equal entitlements.}}
\citet{steinhaus1948problem} initiated fair allocation with the cake-cutting problem, followed by \citet{foley1966resource} advocacy for envy-free resource allocation. Challenges with indivisible items were outlined by \citet{schmeidler1971fair}. 

\paragraph{\textbf{Binary valuations.}}
Particularly relevant to our work is a recent work by \citet{viswanathan2022yankee}, who devised a fair allocation method inspired by Yankee Swap, achieving efficient and fair allocations when agents have submodular \emph{binary} valuations. We use some of their techniques in our algorithms.

\paragraph{\textbf{Subsidies.}}
The concept of compensating an indivisible resource allocation with money has been explored in the literature ever since 
\citet{demange1986multi} introduced an ascending auction for envy-free allocation using monetary payments for unit demand agents.

In addition to their results for additive valuations mentioned earlier,  \citet{brustle2020one} also study the more general class of monotone valuations. They demonstrate that a total subsidy of $2(n-1)^2   V$ is sufficient to guarantee the existence of an envy-freeable allocation.
\citet{kawase2024towards} improved this bound to $\frac{n^2 - n - 1}{2}$.

\citet{caragiannis2021computing} developed an algorithm that approximates the minimum subsidies with any required accuracy for a constant number of agents, though with increased running time. However, for a super-constant number of agents, they showed that minimizing subsidies for envy-freeness is both hard to compute exactly and difficult to approximate.

\citet{aziz2021achieving} presented a sufficient condition and an algorithm to achieve envy-freeness and equitability (every agent should get the same utility) when monetary transfers are allowed for agents with quasi-linear utilities and superadditive valuations (positive or negative).

\citet{barman2022achieving} studied agents with  dichotomous valuations (agents whose marginal value for any good is either zero or one), without any additivity requirement.
They proved that, for $n$ agents, there exists an allocation that achieves envy-freeness with total required subsidy of at most $n-1$, which is tight even for additive valuations.

\citet{goko2024fair} study an algorithm for an envy-free allocation with subsidy, that is also \emph{truthful}, when agents have submodular \emph{binary} valuations. The subsidy per agent is at most $V=1$.
Their algorithm works only for agents with equal entitlements.
% \erel{Check if our truthful mechanism can be extended to submodular-binary valuations.}

\paragraph{\textbf{Different entitlements.}}
In the past few years, several researchers have examined a more general model in which different agents may have different entitlements, included weighted fairness models like \emph{weighted maximin share fairness (WMMS)} and \emph{weighted proportionality up to one item (WPROP1)} (\citet{chakraborty2021picking,babaioff2023fair, aziz2019weighted}).
\citet{mithun2021weighted} established \emph{maximum weighted Nash welfare (MWNW)} satisfies \emph{Pareto optimality} and introduced a weighted extension of EF1. \citet{suksompong2022maximum} demonstrated MWNW properties under binary additive valuations and its polynomial-time computability. They further extended these findings to various valuation types. 

\paragraph{\textbf{Different entitlements with subsidies.}}
\citet{wu2023one} presented a polynomial-time algorithm for computing 
a PROP allocation of \emph{chores} among agents with additive valuations, with total subsidy at most $\frac{n  V}{4}$, which is tight.
For agents with different entitlements, they compute a WPROP allocation with total subsidy at most $\frac{(n-1)  V}{2}$. 
In a subsequent work (\citet{wu2024tree}), they further improved this bound to $(\frac{n}{3} - \frac{1}{6})V$. 

As far as we know, weighted envy-freeness with subsides has not been studied yet. Our paper aims to fill this gap.

\paragraph{\textbf{Matroid-rank valuations.}}
% \hl{new paragraph}
Recent studies have considered \emph{matroid-rank} valuation (binary submodular). \citet{montanari2024weighted} introduce a new family of weighted envy-freeness notions based on the concept of \emph{transferability} and provide an algorithm for computing transferable allocations that maximize welfare. \citet{babaioff2021fair} design truthful allocation mechanisms that maximize welfare and are fair. Particularly relevant to our work is a recent work by \citet{viswanathan2022yankee}, who devised a fair allocation method inspired by Yankee Swap, achieving efficient and fair allocations when agents have submodular binary valuations. We use some of their techniques in our algorithms. Later, \citet{viswanathan2023general} generalize the Yankee Swap algorithm to efficiently compute allocations that maximize any fairness objective, called General Yankee Swap. 
We show that, under binary additive valuations, the general Yankee Swap algorithm can be adapted to agents with different entitlements. 
%The algorithm generates a WEF-able allocation with total required subsidy of at most $W-1$.
%While this is the same as the bound for general additive valuations, we believe the adaptation of Yankee Swap is of independent interest. 
The adaptation is shown in Section \ref{sec:binary-additive}.

\subsection{Our Results} 
We derive bounds on the amount of subsidy required in order to attain a WEF allocation, in several different settings.
We assume, without loss of generality, that the entitlements are ordered in increasing order, i.e., $w_1 \leq w_2 \leq \ldots \leq w_n$.\\
We denote $W := \sum_{i=1}^n w_i$.

In Section \ref{sub:given-allocation}, we assume that an allocation is given.
As shown in Example \ref{example_intro}, there are  instances in which no rearrangement of bundles yields a WEF-able allocation.
We prove a necessary and sufficient condition under which the allocation is WEF-able. We show that, when the allocation is WEF-able, a total subsidy of $(\frac{W}{w_1}-1)mV$ is sufficient to make it WEF, and prove that this bound is tight in the worst case.

This raises the question of whether a weighted-envy-free allocation with subsidy always exists? 
We provide an affirmative answer in Section \ref{sub:computing-wefable-allocation}. 
For agents with additive valuations, assuming all weights are integers, we show an upper bound that is independent of $m$: it is $\frac{W - w_1}{\gcdw}V$,
where $\gcdw$ is the greatest common divisor of all weights --- largest number $d$ such that $w_i/d$ is an integer for all $i\in N$ .
%We show how to compute for agents with integer weights, a WEF-able allocation with total subsidy at most
%$\left(W - w_1\right)$.
%We also prove that this bound cannot be improved in the worst case.

Our algorithm extends \citet{brustle2020one} algorithm, which in the  unweighted setting attains total subsidy at most $(n-1)  V $. 
With equal entitlements, we can normalize the weights to 1, maintaining the same bounds as in the unweighted case.
% \noga{This is important to note that $(n-1)< w_1(n-1) \leq (W-w_1) $}
%$= (\frac{nw_1}{w_1}-1)  V = (\frac{W}{w_1}-1)  V$.

Following \citet{halpern2019fair}, 
in addition to the setting of general additive valuations, we study the setting in which agents have \emph{identical} additive valuations (Section \ref{sec:identical-additive}), and the setting in which agents have  \emph{binary} additive valuations (Section \ref{sec:binary-additive}). 

For identical additive valuations, we compute a WEF-able and WEF$(0,1)$ allocation with total subsidy at most $(n-1)V$, which is tight even in the unweighted case.
%we demonstrate an efficient computation of a WEF allocation with subsidies capped at $V$ for each agent, and a total subsidy bounded by $(n-1)V$,
%which is identical to the unweighted setting.
%We also prove that this bound is tight in the worst case.

Interestingly, in this special case (in contrast with the general case), the bound on the subsidy does not depend on the weights.

In Section \ref{sec:binary-additive}, under binary-additive valuations, %we show how to modify the General Yankee Swap algorithm of \cite{viswanathan2023general} to efficiently compute a WEF-able allocation with subsidies not exceeding $\frac{w_i}{w_1}$ per agent $i \in N$, and a total subsidy bounded by $\frac{W - w_{1}}{w_1} = \frac{W}{w_1} - 1$.
%This bound reduces to an identical overall subsidy of $n-1$ in the unweighted setting.
we adapt the General Yankee Swap algorithm \cite{viswanathan2023general} to compute a WEF-able and WEF$(0,1)$ allocation with total subsidy at most $\frac{W}{w_1} - 1$, reducing to $n-1$ for equal weight.

Our findings and contributions are briefly summarized in Table ~\ref{tab:Table1}.
\begin{table}[hbt!]
	\caption{Distinctions between outcomes established in prior research (see citation), and those newly established in the present study, highlighted in \textbf{bold}.
        \\
        All subsidy upper bounds are attainable by polynomial-time algorithms. 
        \\
        In the unweighted setting $w_1 = \ldots = w_n = 1$, and all upper bounds become $(n-1)V$.
        }
	\label{tab:Table1}
	\begin{tabular}{|c|c|c|c|c|}\toprule
		& \multirow{2}{*}{\shortstack{\textbf{Unweighted}\\ \textbf{Setup}}} & \multicolumn{3}{|c|}{\textbf{Weighted Setup}}  \\
           \cline{3-5}
            & &\makecell{ \textbf{General}\\\textbf{Valuations}}  & \makecell{\textbf{Identical} \\ \textbf{Valuations}}  & \makecell{\textbf{Binary} \\ \textbf{Valuations} } \\ \midrule
		\makecell{
           \textbf{Character-}\\\textbf{ization} \\\textbf{of WEF} \\\textbf{Allocation}} & \makecell{(1) No \\ positive \\cost cycles,\\ (2) USW\\ maximi-\\zation 
           \cite{halpern2019fair}}& \multicolumn{3}{|c|}{\makecell{No positive-cost cycles. \\ (\textbf{\ref{thm:wefable--iff-no-cycles}})}} \\ \midrule
		\makecell{\textbf{Permut-}\\ \textbf{ation of }\\ \textbf{a given}
            \\
            \textbf{allocation,} \\ \textbf{that} \\ \textbf{maximizes} \\ \textbf{sum of} \\ \textbf{values}
            }& \makecell{Always \\ EF-able \\ \cite{halpern2019fair}}& \makecell{
            Not \\necessarily \\WEF-able \\ (\textbf{\ref{MWUSEneqWEF}})}& \makecell{Always\\ WEF-able.\\ (\textbf{\ref{thm:wefable--iff-no-cycles1}})} &\makecell{For non- \\redundant\\allocation:\\ Always \\ WEF-able \\ (\textbf{\ref{non-redundant is WEF-able}})}\\ \midrule
		\makecell{\textbf{Total} \\ \textbf{subsidy}\\
            \textbf{upper} \\ \textbf{bound}} & \makecell{$(n-1)V$ \\ \cite{halpern2019fair}} & \makecell{
            $\frac{W-w_1}{\gcdw}V$
            \\
            (\textbf{\ref{cor: sub general additive}})
            }
            &
            \makecell{ $(n-1)V$
            \\ (\textbf{\ref{thm:wefable--iff-no-cycles3}})}&\makecell{
            $\frac{W}{w_1} - 1$
            \\ (\textbf{\ref{app:theorem_21}})}  \\ \midrule
            \makecell{\textbf{Subsidy}\\
            \textbf{bound} \\ \textbf{for a} \\ \textbf{given} \\ \textbf{allocation}} & \makecell{$(n-1)mV$ \\ \cite{halpern2019fair}} & \multicolumn{3}{|c|}{\makecell{
            $\left(\frac{W}{w_1}-1\right)mV$
            \\
            (\textbf{\ref{worst case allocation is given with weights}})
            }}
            \\ \midrule
            
            \makecell{\textbf{Total}\\\textbf{subsidy}\\
            \textbf{lower} \\ \textbf{bound}} & \makecell{$(n-1)V$ \\ \cite{halpern2019fair}} & \makecell{
            $\left(\frac{W}{w_1}-1\right)V$
            \\
            (\textbf{\ref{worst case allocation can be chosen with weights}})
            } & \makecell{$(n-1)V$ \\ (\textbf{\ref{identical-additive-tightness}})} & \makecell{$\frac{W}{w_2} - 1$ \\ (\textbf{\ref{prop:lower-bound-binary})}}
            \\ 
            
		\bottomrule
	\end{tabular}
\end{table}
%Due to space constraints, all proofs are deferred to appendices.
\begin{remark}
    Practical fair allocation cases use budget-balanced payments rather than subsidies. We use the subsidies terminology for consistency with previous works.
	%We note that our subsidy upper bounds imply upper bounds on the maximum payment per agent in the budget-balanced model; see Appendix \ref{sec:models} for the precise statements.
\end{remark}
\section{Preliminaries}

\paragraph{\textbf{Agents and valuations.}}

We denote by $[t]$ the set $\{1, 2, ..., t\}$ for any positive integer $t$. We focus on the problem of allocating $M = \{o_{1},...,o_{m}\}$ indivisible items among  $N = [n]$ agents. Each subset of $M$ is called a \textit{bundle}, 
and a combination of a bundle and a monetary transfer (positive or negative) is called a \emph{portion}.
% \erel{We need a term for a subset of objects + money.}
% \rica{how about compensated bundle? Should we use portion as use suggested?}

Each agent $i \in N$ has a \textit{valuation function} $v_{i} : 2^{M} \rightarrow \mathbb{R}_{\geq 0}$, indicating how much they value different bundles. For simplicity, we write $v_{i}(o_1 ,...,o_t)$ instead of $v_{i}(\{o_1, ..., o_t\})$.
% for any $i \in N$ and $\{o_1, ..., o_t\} \subseteq M$. 
Additionally, we define $V = \displaystyle\max_{i\in N, o\in M} v_i(o)$.
\iffalse
Let $j \in \displaystyle\argmaA_{i\in N, o\in N} v_i(o)$. We define $V_{2nd}$ to be the second highest value, that is, $V_{2nd} = \displaystyle\maA_{i \neq j\in N, o\in M} v_i(o)$.
\fi
% \erel{
% What if the second-highest is $v_i(o')$, i.e., same agent but different item?
% } \noga{We use $V_{2nd}$ for bound the envy, so we need to define it over another agent.}

The set $\Pi_{k}(S)$ denotes the collection of ordered partitions of a set $S \subseteq M$ into $k$ bundles.
An \textit{allocation}, denoted as $A \in \Pi_{n}(M)$, assigns the items to the agents. It consists of $n$ disjoint bundles $(A_i)_{i \in N}$, where $\forall {i\in N}: A_{i} \subseteq M$, and $A_i \cap A_j = \emptyset$ for all $i \neq j \in N$. The bundle $A_i$ is given to agent $i \in N$, and $v_i(A_i)$ represents agent $i$'s valuation of their bundle under allocation $A$. 
We consider only allocations $A$ that are \textit{complete}, that is, $\cup_{i\in N} A_{i} = M$.

% \erel{I suggest to mention subsidies already here.}
% \rica{added including uniting subsidies and payments.}
\iffalse
In addition to bundle $A_i$, agent $i \in N$ may receive a subsidy $s_i$
where a vector of subsidies is denoted as $\mathbf{s} = (s_1 ,..., s_n)$, such that $s_i \geq 0$ (subsidy) or $s_i < 0$ (payment). 
\erel{I think we do not use payments now.}
\fi
We assume that agents are \emph{quasilinear}, so that the utility of each agent is $$u_i(A_i,s_i) := v_{i}(A_{i}) + s_{i}.$$

We assume that the valuation functions $v_i$ are \textit{normalized}, that is $v_i(\emptyset) = 0$ for all $ i \in N$ (\citet{suksompong2023weighted});
and \textit{additive}, that is $v_{i}(S) = \sum_{o\in S} v_i(o)$ for all $S \subseteq M , i \in N$.
In some sections of this paper, we make some additional assumptions:
\begin{enumerate}
\item \textit{identical additive}: there exists an additive valuation function $v$ such that $v_i\equiv v$ for all $i \in N$.
\item \textit{binary additive}: for all $i \in N$ and $o \in M$, $v_{i}(o) \in \{0,1\}$.
\end{enumerate}
Without loss of generality, we assume that each item is valued positively by at least one agent; items that are valued at $0$ by all agents can be allocated arbitrarily without affecting the envy.

\paragraph{\textbf{Entitlements.}}
Each agent $i \in N$ is endowed with a fixed \emph{entitlement}  $w_i \in \mathbb{R}_{>0}$. We also refer to entitlement as \emph{weight}.
% The entitlement does not affect the agent's utility; rather, a player's envy towards another player arises from comparing the allocation of resources between them, taking into account their different entitlements.
We assume, without loss of generality, that the entitlements are ordered in increasing order, i.e., $w_1 \leq w_2 \leq \ldots \leq w_n$. 
We denote 
$W := \sum_{i\in N} w_i$. 
%Some of our results also assume that the entitlements are integers, $w_i \in \mathbb{Z}_{>0}$. \noga{I think we should mention it only where it used}

\paragraph{\textbf{Utilitarian social welfare.}}

The \textit{utilitarian social welfare (USW)} of an allocation $A \in \Pi_n(M)$ is $USW(A) = \sum_{i\in N} v_{i}(A_{i})$. An allocation $A \in \Pi_n(M)$ is called \emph{maximizing utilitarian social welfare}, denoted \textit{$MUSW^{M,N,v}$}, if and only if $\sum_{i\in N}v_{i}(A_{i})\geq \sum_{i\in N}v_{i}(B_{i})$ for any other allocation $B \in \Pi_n(M)$.

The \textit{weighted utilitarian social welfare (WUSW)} of an allocation $A \in \Pi_n(M)$ is the sum of values obtained by each agent, scaled by ratio of their weights, denoted as $WUSW(A) = \sum_{i\in N} w_i   v_{i}(A_{i})$. 
% This approach maximizes overall satisfaction while considering entitlements, giving more importance to agents with higher entitlements.
An allocation $A \in \Pi_n(M)$ is considered \textit{$MWUSW^{M,N,v}$} if and only if $\sum_{i\in N} w_i   v_{i}(A_{i})\geq \sum_{i\in N} w_i   v_{i}(B_{i})$ for any other allocation $B \in \Pi_n(M)$.

To see why we multiply by the weights, consider a setting with one item and two agents, who value the item at $v_1 = 4$ and $v_2 = 6$.
With equal weights, the MUSW allocation naturally gives the item to agent 2, who values it higher. However, if the weights are e.g. $w_1=3,w_2=1$, the MWUSW allocation gives the item to agent 1, since the weighted value $w_1  v_1 = 12$ is larger  than $w_2  v_2 = 6$. Multiplying the values by the weights gives higher priority to the agent with the higher weight, as expected.
\paragraph{\textbf{Envy.}}
The \textit{envy} of agent $i$ towards agent $j$ under an allocation $A$ and subsidy vector $\mathbf{s}$ is defined as 
\[u_i(A_j) - u_i(A_i) = 
(v_i(A_j)+s_j) - (v_i(A_i)+s_i).\]
If for all $i,j\in N$ the envy of $i$ towards $j$ is at most $0$, then $(A,\mathbf{s})$ is called \textit{envy-free (EF)}.

The \textit{weighted-envy} of agent $i$ towards agent $j$ under an allocation $A$ and subsidy $\mathbf{s}$ is defined as 
\[\frac{u_i(A_j)}{w_j} - \frac{u_i(A_i)}{w_i}
= \frac{v_i(A_j)+s_j}{w_j} - \frac{u_i(A_i)+s_i}{w_i}.\]
If for all $i,j\in N$ the weighted envy of $i$ towards $j$ is at most $0$, then $(A,\mathbf{s})$ is called a \textit{weighted-envy-free (WEF) solution}.

To see why we \emph{divide} by the weights, 
consider again a setting of one item and two agents who value it identically, $v_1=v_2=V$, and suppose the item is given to  agent 1, and agent 2 gets some subsidy $s_2$.
With equal entitlements, the envy of agent 2 is $(v_2) - (s_2)$, and a subsidy of $s_2=V$ is required to eliminate this envy.
However, if the weights are $w_1=3,w_2=1$, 
the envy of agent 2 is $(v_2/w_1) - (s_2/w_2)$, and the required subsidy is only $s_2=V/3$; 
In contrast, if the weights are $w_1=1,w_2=3$, then the required subsidy is $s_2=3V$.
This is expected, as there is less justification for envy when agent 1 has a higher entitlement, and more justification for envy when agent 2 has a higher entitlement.

Intuitively,
the term $\frac{v_i(A_i)}{w_{i}}$ represents the value per unit entitlement for agent $i$ in their allocation. The WEF condition ensures that this value is at least as high as $\frac{v_i(A_j)}{w_{j}}$, denoting the corresponding value per unit entitlement for agent $j$ in the same allocation.

% Dividing the value that agent $i$ assigns to the bundle of agent $j$ by $w_j$ normalizes the valuation according to $j$'s entitlement. This is necessary to ensure fairness, as agents with higher entitlements should naturally receive more, and envy must be assessed proportionately. An allocation $A$ is said to be \textit{weighted-envy-free (WEF)} if and only if $\frac{v_i(A_i)}{w{i}} \geq \frac{v_{i}(A_{j})}{w_{j}}$ for all agents $i,j \in N$.

WEF seamlessly reduces to envy-free (EF) concept when entitlements are equal, i.e., $w_i = w_j$ for all $i,j \in N$. Like EF, which can be challenging to achieve with indivisible items, WEF faces similar limitations in guaranteeing fairness under such circumstances.

Often, we are given only an allocation $A$, and want to find a subsidy vector with which the allocation will be WEF.
\begin{definition}
(a)
We say that subsidy vector $\mathbf{s}$ is \textit{weighted-envy-eliminating} for allocation $A$ if $(A,\mathbf{s})$ is WEF.

(b)
An allocation $(A_{i})_{i \in N}$ is called \textit{weighted-envy-freeable (WEF-able)} if there exists a subsidies vector $\mathbf{s}$ such that $(A,\mathbf{s})$
is WEF.
\end{definition}

In the setting without money, WEF can be relaxed by allowing envy up to an upper bound that depends on the item values. We employ the generalization of allowable envy proposed by \citet{chakraborty2022weighted}:
\begin{definition} (\citet{chakraborty2022weighted}). For $x,y \in [0,1]$, an allocation $A$ is said to satisfy \emph{WEF(x,y)} if for any $i,j \in N$, there exists $B\subseteq A_j$ with $|B| \leq 1$ such that $$\frac{v_i(A_i) + y  v_i(B)}{w_i} \geq \frac{v_i(A_j) - x  v_i(B)}{w_j}.$$
\end{definition}
%This definition implies that the permissible level of envy is the weighted average of $\frac{v_i(B)}{w_i}$ and $\frac{v_i(B)}{w_j}$. 
Specifically, \emph{WEF(0,0)} corresponds to \emph{WEF}.
% while \emph{WEF(1,0)} is equivalent to the notion of \emph{WEF1} proposed by Chakraborty et al. \cite{mithun2021weighted}.

\paragraph{\textbf{Weighted envy graph.}}
The \emph{weighted envy graph} of allocation $A$, denoted $G_{A,w}$, is a complete directed graph consisting of a set of vertices representing agents $N$, each assigned a weight denoted by $w_{i}$. In case of identical weights, we denote it  $G_{A}$.

For any pair of agents $i$ and $j$ in $N$, the cost assigned to the arc $(i,j)$ in $G_{A,w}$ is defined as the weighted envy that agent $i$ holds toward agent $j$ under allocation $A$:
\[
 cost_A(i,j) := \frac{v_i(A_j)}{w_j} - \frac{v_i(A_i)}{w_i}.\]

We denote the cost of a path $(i_1,...,i_k)$ as $cost_A(i_1,...,i_k) = \sum_{j = 1}^{k-1} cost_A(i_j, i_{j+1})$.
With these definitions, $\ell_{i,j}(A)$ 
 represents the maximum-cost path from $i$ to $j$ in $G_{A,w}$, and $\ell_i(A)$ represents the maximum-cost path in $G_{A,w}$ starting at $i$.

The previous work of \citet{halpern2019fair} provides sufficient and necessary conditions for an EF-able allocation  in the unweighted setup. The following theorem considers a pre-determined allocation $A =  (A_1,\ldots,A_n)$, where we aim to assign one bundle to each agent. 

\begin{theorem} 
\label{theorem halpern and shah}
(\citet{halpern2019fair}). 
In a setting with equal entitlements, the following statements are equivalent:
\begin{enumerate}
\item The allocation $A$ is envy-freeable.
\item The allocation $A$ maximizes the utilitarian welfare across all reassignments of its bundles to agents, that is, for every permutation $\sigma$ of $[n]$ (a bijection $\sigma : [n] \rightarrow [n]$), $\sum_{i\in N}v_{i}(A_{i}) \geq \sum_{i\in N} v_{i}(A_{\sigma{i}})$. \label{part b of theorem halpern and shah}
\item $G_{A}$ has no positive-cost cycles. \label{part c of theorem halpern and shah}
\end{enumerate}
\end{theorem}
The theorem implies that,  in the unweighted setup, for every allocation there exists a reassignment of bundles allocated to the agents, ensuring its envy-freeability.

\citet{halpern2019fair} have outlined a method to find a subsidy vector $\mathbf{s}$%\footnote{Halpern and Shah employ payment notions that align with the subsidy model.}
, ensuring envy-freeness while minimizing the required subsidy. Each agent's subsidy is determined by the maximum cost path from that agent in the envy graph. This cost can be computed within strongly polynomial time.

Furthermore, they identify the subsidy needed in the worst case for two scenarios: 
\begin{enumerate}
    \item \textbf{When the allocation is given:} the minimum subsidy required is $(n-1)  m   V$ in the worst case \cite{halpern2019fair}.
    \item \textbf{When the allocation can be chosen:} the minimum subsidy required is at least $(n-1)   V$ in the worst case, even in the special case of binary valuations \cite{halpern2019fair}.
\end{enumerate}
%when the allocation is predetermined and when it can be chosen.

\section{WEF-able Allocations}\label{section 3:section_3}
In this section, we expand the idea of fairness by considering weights and discuss the challenge of identifying WEF-able allocations, compared to EF-able allocations. 

We demonstrate how to compute a subsidy vector that eliminates envy among the agents for a given allocation.
Naturally, the question arises whether a WEF-able allocation for agents with general additive valuations always exists. We address and answer this inquiry.

%The weighted-envy-free concept is delineated within the graph $G_{A,w}$ through the conspicuous absence of edges characterized by positive costs. At the same time, the state of being weighted-envy-freeable aligns with the augmentation of agents' profits, a condition met by preserving edges featuring non-positive costs. This dichotomy underscores the significance of the elimination of positive-cost edges in attaining allocations that are weighted-envy-free. It further emphasizes the critical role played by the retention of non-positive costs in facilitating the overall enhancement of agents' profits. 
% \noga{I think that we can remove this paragraph}
% \rica{I agree. Done.}

\subsection{A given allocation: the weighted-envy graph}
\label{sub:given-allocation}

 When trying to extend \Cref{theorem halpern and shah} to the weighted setup, a notable distinction arises: in the unweighted context, an allocation $A$ is EF-able if and only if $A$  maximizes the utilitarian welfare across all reassignments of its bundles to agents (part (\ref{part b of theorem halpern and shah}) of
\Cref{theorem halpern and shah}). 

Formally, for every permutation $\sigma$ of $[n]$, it holds that $$\sum_{i\in N}( v_{i}(A_{i})) \geq \sum_{i \in N}(v_{i}(A_{\sigma(i)})).$$

Contrastingly, in the weighted setup, this assertion does not necessarily hold, as demonstrated by Example \ref{example_intro}.
Moreover, the converse of (\ref{part c of theorem halpern and shah}) in \Cref{theorem halpern and shah} can fail too. 
\begin{proposition} 
\label{theorem wefable is not usw}
    There exists an allocation $A$, which is WEF-able, yet maximizes neither the utilitarian social welfare or the weighted utilitarian social welfare.
\end{proposition}
\begin{proof}
    Consider two agents with weights $w_{1} = 2, w_{2} = 3$ and $2$ items. The valuation functions are as follows:
\
\[
\begin{bmatrix}
   & o_1 & o_2 \\
  i_1 & 8 & 10 \\
  i_2 & 6 & 7 \\
\end{bmatrix}
\]
Consider the allocation $A=(A_{1}, A_{2})$ when $A_{1}=\{o_{1}\}, A_{2}=\{o_{2}\}$, and subsidies $s_{1} = 0, s_{2}= 2$: 
\begin{enumerate}
\item The envy of agent $i_1$ towards agent $i_2$ is $\frac{v_{1}(A_{2}) + s_{2}}{w_2} - \frac{v_{1}(A_{1}) + s_{1} }{w_1}= \frac{10 + 2}{3} - \frac{8 + 0}{2} = 0$
\item the envy of agent $i_2$ towards agent $i_1$ is $\frac{v_{2}(A_{1})+ s_{1} }{w_1} -\frac{v_{2}(A_{2}) + s_{2}}{w_2} = \frac{6 + 0}{2}  - \frac{7 + 2}{3} = 0$
\end{enumerate}
so the allocation is WEF. However, we can achieve both higher USW and higher WUSW by swapping the bundles: $A_{1}^{'} = A_{2}, A_{2}^{'} = A_{1}$, as
\begin{align*}
&v_{1}(A_{1}) + v_{2}(A_{2}) = 8 + 7 = 15 < \\&16 = 10 + 6 = v_{1}(A_{1}') + v_{2}(A_{2}')
\\
&w_1   v_{1}(A_{1}) + w_2   v_{2}(A_{2}) = 16 + 21 = 37 < \\&38 = 20 + 18 = w_1   v_{1}(A_{1}') + w_2   v_{2}(A_{2}'). \qedhere
\end{align*}
\end{proof}

It turns out that two out of three parts of this theorem still hold in the weighted setting; The proof is similar to the one in \cite{halpern2019fair}.
\begin{theorem} 
\label{thm:wefable--iff-no-cycles}
    For an allocation $A$, the following statements are equivalent:
    \begin{enumerate} [label=(\alph*)]
    \item $A$ is weighted-envy-freeable \label{condition_a_theroem_1}
    \item $G_{A,w}$ has no positive-cost cycles. \label{condition_b_theroem_1}
\end{enumerate}
\end{theorem} 
\begin{proof}
We show $(a) \Rightarrow (b)$ and $(b) \Rightarrow (a)$.

For $(a) \Rightarrow (b)$,
we prove a more general claim: the total cost of any cycle in $G_{A,w}$ is the same for any subsidy vector.
Indeed, suppose we give some agent $i$ a subsidy of $s_i$. As a result, the cost of every edge from $i$ decreases by $s_i / w_i$ (as $i$ experiences less envy), and the cost of every edge into $i$ increases by $s_i / w_i$ (as other agents experience more envy in $i$).
Every cycle through $i$ contains exactly one edge from $i$ and one edge into $i$, and every other cycle contains no such edges. Therefore, the total cost of any cycle
remains unchanged.

Now, if $A$ is WEF-able, then for some subsidy vector, the costs of all edges are at most $0$, so the total costs of all cycles are at most $0$. Therefore, the costs of all cycles are at most $0$ even without the subsidy vector.
% If there is a positive cost cycle in $G_{A,w}$, then the same cycle has the same positive cost for any subsidy vector. Therefore, for any subsidy vector, at least one edge in that cycle has a positive cost, which means that there is envy along that edge. Therefore, $A$ is not WEF-able.

For $(b) \Rightarrow (a)$, we present a specific subsidy vector.
Suppose $G_{A,w}$ has no positive-cost cycles. Then, $\ell_i(A)$ (the maximum cost of a path in $G_{A,w}$ starting at $i$) is well-defined and finite. For each $i \in N $, let $s_i = w_{i}   cost_{A}(\ell_i(A))$. 
It is noteworthy that $s_{i} \geq 0$, since there is a path from $i$ to $i$ with $0$ cost, establishing the suitability of \textbf{s} as a valid subsidies vector.

Furthermore, the following holds for all pairs $i \neq j \in N$:
% (where $\ell_{i,j}(A)$ denotes the highest-cost path from $i$ to $j$ in $G_{A,w}$):
\begin{align*}
\frac{s_i}{w_i} 
&= cost_{A}(\ell_i(A)) \geq cost_{A}(i,j) + cost_{A}(\ell_j(A)) 
\\
&= \frac{v_{i}(A_{j})}{w_{j}} - \frac{v_{i}(A_{i})}{w_{i}} + cost_{A}(\ell_j(A)) 
\\
&= \frac{v_{i}(A_{j})}{w_{j}} - \frac{v_{i}(A_{i})}{w_{i}} + \frac{s_j}{w_j}  .
\end{align*}
  This implies $\frac{v_{i}(A_{i})+ s_{i}}{w_{i}}  \geq \frac{v_{i}(A_{j}) + s_{j}}{w_{j}}$.
Hence, $(A,\mathbf{s})$ is weighted-envy-free, and thus, $A$ is weighted-envy-freeable.
\end{proof}
\Cref{thm:wefable--iff-no-cycles} presents an effective method for verifying whether a given allocation $A$ is WEF-able.
\begin{proposition}
Given an allocation $A$, it can be checked in polynomial time whether $A$ is WEF-able.
\end{proposition}
\begin{proof}
    According to \Cref{thm:wefable--iff-no-cycles}, determining whether $A$ is is WEF-able is equivalent to verifying the absence of positive-cost cycles in the weighted envy-graph $G_{A,w}$. This can be achieved by transforming the graph by negating all edge weights and then checking for the presence of negative-cost cycles. Using the Floyd-Warshall algorithm 
    (\citet{weisstein2008floyd, wimmer2017floyd}) on the graph obtained by negating all edge cost in $G_{A,w}$ requires $O(n^3)$ time. Constructing the initial graph $G_{A,w}$ takes $O(mn)$ time, resulting in an overall complexity of $O(mn + n^3)$.   
\end{proof}

We use \Cref{thm:wefable--iff-no-cycles} to present an alternative proof that part 2 of \Cref{theorem halpern and shah} does not hold in the weighted setting.

\begin{proposition} 
\label{MWUSEneqWEF}
There exist bundles $B_1, ..., B_n$ 
such that for every permutation of the agents $\sigma : [n] \rightarrow [n]$, 
the resulting allocation $A_i = B_{\sigma(i)}$ is not WEF-able.
% \footnote{This Proposition holds true even when maximizing weighted utilitarian welfare.}
%In particular, the permutations that maximize the social welfare or the weighted social welfare are not WEF-able.
\end{proposition}
\begin{proof}
As in Example \ref{example_intro}, consider 2 bundles $B_1, B_2$ and 2 agents, $N=\{i_{1}, i_{2}\}$, 
with weights $w_{1} = 1, w_{2} = 10$ and valuation functions 
\[
\begin{bmatrix}
   & B_1 & B_2 \\
  i_1 & 5 & 7 \\
  i_2 & 10 & 8 \\
\end{bmatrix}
\]
We show that for every permutation of the agents $\sigma :[n] \rightarrow [n]$, the resulting allocation $A_i = B_{\sigma(i)}$ results in the corresponding weighted envy graph a positive cost cycle.

The allocation $A_{1} = B_1 , A_{2} = B_2$ maximizes both the utilitarian social welfare and the weighted utilitarian social welfare.
The weighted envy of agent $i_1$ towards agent $i_2$ is $cost_{A}(i_{1}, i_{2}) = \frac{v_{i_{1}}(A_{i_{2}})}{w_{i_{2}}} - \frac{v_{i_{1}}(A_{i_{1}})}{w_{i_{1}}} = \frac{7}{10} - 5 = -\frac{43}{10}$;
the weighted envy of agent $i_2$ towards agent $i_1$ is calculated as $cost_{A}(i_{2}, i_{1}) = \frac{v_{i_{2}}(A_{i_{1}})}{w_{i_{1}}} - \frac{v_{i_{2}}(A_{i_{2}})}{w_{i_{2}}} = 10 - \frac{8}{10} = \frac{92}{10}$.
Thus, the cost of the cycle $(i_{1}, i_{2})$ is $\frac{49}{10} > 0$, indicating a positive-cost cycle. \\

Also, if an attempt is made to swap the bundles in a circular manner, such that $A_{1} =A_2, A_{2} = B_1$, the cost of this cycle $(i_{1}, i_{2})$ is $(\frac{5}{10} - 7) + (8 - \frac{10}{10}) = \frac{1}{2} > 0$, affirming a positive cost as well. 

By \Cref{thm:wefable--iff-no-cycles}, both permutations are not WEF-able.
\end{proof}

Deciding that an allocation $A$ is WEF-able is not enough. As in \citet{halpern2019fair}, we also need to find a minimal subsidy vector, $\mathbf{s}$, that ensures WEF for $(A, \mathbf{s})$. This subsidy vector can be computed in strongly polynomial time and matches the one constructed in the proof of \Cref{thm:wefable--iff-no-cycles}. 
\begin{theorem} \label{max_path_subsidy}
For any allocation $A$ that is WEF-able, let $\mathbf{s}^{*}$ be a subsidy vector defined by 
$$\mathbf{s}^{*}_{i} := w_i \cdot cost_{A}(\ell_i(A)),$$
for all $i\in N$. 
Then
\begin{enumerate}
    \item $(A,\mathbf{s}^{*})$ is WEF;
    \item Any other envy-eliminating  subsidy vector $\mathbf{s}$
    % , ensuring $(A, \mathbf{s})$ is WEF, 
    satisfies $s^{*}_{i} \leq s_{i}$ for all $i\in N$;
    \item The computation of $s^{*}$ can be performed in $O(nm+n^{3})$ time.
    \item There exists at least one agent $i\in N$ for whom $s^*_i=0$. \label{agent with 0 subsidy}
\end{enumerate}
\end{theorem}
\begin{proof}
\begin{enumerate}
    \item The establishment of condition \ref{condition_b_theroem_1} implying condition \ref{condition_a_theroem_1} in Theorem  \ref{thm:wefable--iff-no-cycles} has already demonstrated the inclusion of $(A,\mathbf{s}^{*})$ is WEF. 
    \item Let $\mathbf{s}$ be a subsidy vector, such that $(A,\mathbf{s})$ is WEF, and $i \in N$ be fixed. Consider the highest-cost path originating from $i$ in the graph $G_{A,w}$. This path is denoted as $(i_{1}, ..., i_{k})$, with $i_{1} = i$ and $cost_{A}(i_{1}, ..., i_{r}) = cost_{A}(\ell_i(A)) = \frac{s^{*}_{i}}{w_i}$.
   
    Due to the WEF nature of $(A,\mathbf{s})$, it follows that for each $k \in [r-1]$, the following inequality holds:
    \begin{align*}
        & \frac{v_{i_{k}}(A_{i_{k}})+ s_{i_{k}}}{w_{i_{k}}}  \geq \frac{v_{i_{k}}(A_{i_{k+1}}) + s_{i_{k+1}}}{w_{i_{k+1}}}\\
    & \Rightarrow \frac{s_{i_{k}}}{w_{i_{k}}} - \frac{s_{i_{k+1}}}{w_{i_{k+1}}} \geq \frac{v_{i_{k}}(A_{i_{k+1}})}{w_{i_{k+1}}} - \frac{v_{i_{k}}(A_{i_{k}})}{w_{i_{k}}} = \\ &cost_A(i_{k}, i_{k+1}).
    \end{align*}
    Summing this inequality over all $k \in [r-1]$, the following relation is obtained:
    \begin{align*}
   &\frac{s_{i}}{w_{i}} - \frac{s_{i_{r}}}{w_{i_{r}}} = \frac{s_{i_{1}}}{w_{i_{1}}} - \frac{s_{i_{r}}}{w_{i_{r}}}\geq cost_A(i_{1}, ..., i_{r}) = \frac{s_{i}^{*}}{w_i}
   \\& \Rightarrow \frac{s_{i}}{w_{i}} \geq \frac{s_{i}^{*}}{w_i} + \frac{s_{i_{r}}}{w_{i_{r}}} \geq \frac{s_{i}^{*}}{w_i}.
   \end{align*}
    The final transition is valid due to the non-negativity of subsidies and weights, that is, $\frac{s_{i_{r}}}{w_{i_{r}}} \geq 0$. 
    \item The computation of $s^{*}$ can be executed as follows: Initially, the Floyd-Marshall algorithm 
(\citet{weisstein2008floyd, wimmer2017floyd}) is applied to the graph derived by negating all edge costs in $G_{A,w}$ (This has a linear time solution since there are no cycles with positive costs in the graph). Hence, determining the longest path cost between any two agents, accomplished in $O(nm+n^{3})$ time. Subsequently, the longest path starting at each agent is identified in $O(n^{2})$ time.
    \item Assume, for the sake of contradiction, that $s^*_i > 0$ for every agent $i \in N$, which implies that $\text{cost}_{A}(\ell_i(A)) > 0$ since $w_i \geq 1$. Because the number of agents is finite, there must be $k \geq 2$ agents, say $i_1, \dots, i_k$, such that $\ell_{i_1}(A), \dots, \ell_{i_k}(A)$ form a cycle with positive cost in $G_{w,A}$, contradicting \Cref{thm:wefable--iff-no-cycles}.

Therefore, there must be at least one agent whose subsidy is $0$.
\end{enumerate}
\end{proof}
Now we can find the minimum subsidy needed in the worst-case scenario for agents with different entitlements, whether the allocation is given or can be chosen. The proofs extend those of \citet{halpern2019fair}. 
Note that \Cref{worst case allocation is given with weights} applies to any additive valuations, including identical and binary cases.

\begin{theorem} \label{worst case allocation is given with weights} 
For every weight vector
and every given WEF-able allocation $A$, 
letting $s_i := s^*_i = w_i \ell_i(A)$,
the total subsidy $\displaystyle \sum_{i\in N}s_i$
is at most $\left(\frac{W}{w_1} - 1\right)   m   V$, and this bound is tight in the worst case.
\end{theorem}
\begin{proof}
   The proofs extend those of \cite{halpern2019fair}. 
By \Cref{max_path_subsidy}, to bound the subsidy required for $i$, we bound the highest cost of a path starting at $i$.
We prove that, for every WEF-able allocation $A$ and agent $i$, the highest cost of a path from $i$ in $G_{A,w}$ is at most $\frac{m   V}{w_1}$.

For every path $P$ in $G_{A,w}$,
\begin{align*}
&cost_A(P) =
\sum_{(i,j)\in P} cost_A(i,j)  = \sum_{(i,j)\in P}
\frac{v_{i}(A_{j})}{w_{j}} - \frac{v_{i}(A_{i})}{w_{i}} \\
& \leq \sum_{(i,j)\in P} \frac{v_{i}(A_{j})}{w_1} \leq 
\sum_{(i,j)\in P}\frac{V\cdot |A_{j}|}{w_1}
\leq 
 \frac{m V}{w_1}.  %& \text{(sum of bundle sizes $\leq m$)}
\end{align*} 
Therefore, the cost of every path is at most $\frac{m V}{w_1}$, so agent $i$ needs a subsidy of at most $w_i \frac{m V}{w_1}$. 
By part (4) of \Cref{max_path_subsidy}, at least one agent has a subsidy of 0. This implies a total subsidy of at most $\frac{W - w_1}{w_1}  m   V %= \left(W - 1\right)   m   V.
   = \left(\frac{W}{w_1} - 1\right) m V $.

To establish tightness, consider 
   an instance with $m$ identical items, which all $n$ agents value at $V$.
   Consider the allocation $A$, which assigns all items to a single agent $j$ with the minimum entitlement. It's evident that $A$ is WEF-able, and its optimal subsidy vector $\mathbf{s}$ satisfies $s_j = 0$ and $s_i = \frac{w_i}{w_j}   m   V =  \frac{w_i}{w_{1}}   m   V$ for $i\neq j \in N$. Therefore, we require 
   \[
   \frac{W - w_{1}}{w_1}  m   V %= \left(W - 1\right)   m   V.
   = \left(\frac{W}{w_1} - 1\right) m V .
   \]
\end{proof}

%Importantly, for a fixed allocation, the weighted setting maximum worst-case subsidy depends on the proportion of $w_1$ in the total weight $W$, which can be much larger than the number of agents $n$, the argument that the unweighted setting depends on.
In the unweighted case $W/w_1=n$, so the upper bound on the subsidy becomes $(n-1)mV$.
This is the same upper bound proved by 
\citet{halpern2019fair} for the unweighted case and additive valuations.

Halpern and Shah proved in \cite{halpern2019fair} that, if an allocation $A$ is EF-able and EF1, 
then every path in $G_{A,w}$ has cost at most $(n-1)   V$, and therefore a subsidy of at most $(n-1)^2   V$ is sufficient.
The following theorem generalizes this result to the weighted setting.
\begin{theorem}
    \label{WEF-able and WEF(x,y) subsidy bound}
   Let $A$ be both WEF-able and $WEF(x,y)$ for any $x,y \in [0,1]$. Then there exists an envy-eliminating subsidy vector with total subsidy at most $\left(\frac{W}{w_1} - 1 \right) (n-1)V$.
\end{theorem}
\begin{proof}
    $A$ is WEF$(x,y)$, so for all $i,j\in N$, there exists $B\subseteq A_j$ with $|B|\leq 1$ where
    \begin{align*}
        & \frac{v_i(A_j)}{w_j} - \frac{v_i(A_i)}{w_i} \leq \frac{y  v_i(B)}{w_i} + \frac{x  v_i(B)}{w_j} \le \\
        &\frac{y  v_i(B)}{w_1} + \frac{x  v_i(B)}{w_1} = 
         \frac{(x+y)  v_i(B)}{w_1} \le \\
         &(x+y)\frac{V}{w_1}.
    \end{align*}
    Any path contains at most $n-1$ arcs, that is, $$s_i = w_i \ell_i \leq (x+y)\frac{w_i}{w_1}   (n-1)  V.$$
    From part (\ref{agent with 0 subsidy}) of \Cref{max_path_subsidy}, there is at least one agent that requires no subsidy, so the required total subsidy is at most $$\frac{W - w_1}{w_1} (n-1)  V = \left(\frac{W}{w_1} - 1 \right) (n-1)V.$$
\end{proof}

\subsection{Computing a WEF-able allocation}
\label{sub:computing-wefable-allocation}
We start with a lower bound on the required subsidy for general additive valuations.
To prove the lower bound, we need a Lemma.
\begin{lemma}
    \label{lem:single-item}
Suppose there are $n$ agents,
and only one item $o$ which the agents value positively. Then an allocation is WEF-able iff $o$ is given to an agent $i$ with the highest $v_i(o)$.
\end{lemma}
\begin{proof}
By \Cref{thm:wefable--iff-no-cycles}, it is sufficient to check the cycles in the weighted envy-graph.
If $o$ is given to $i$, then $i$'s envy is $- \frac{v_i(o)}{w_i}$,
and the envy of every other agent $j$ in $i$ is $\frac{v_j(o)}{w_i}$.
All other envies are 0.
The only potential positive-weight cycles are cycles of length 2 involving agent $i$.
The weight of such a cycle is positive iff $\frac{v_j(o)}{w_i} - \frac{v_i(o)}{w_i} >0$, which holds iff $v_j(o) > v_i(o)$.
Therefore, there are no positive-weight cycles iff $v_i(o)$ is maximum.
\end{proof}

\begin{theorem}
\label{worst case allocation can be chosen with weights}
For every weight vector $\mathbf{w}$ and any $n\geq 2$,
no algorithm can guarantee a total subsidy smaller than $\left(\frac{W}{w_1}-1\right)  V$.
\end{theorem}
\begin{proof}
% \erel{Need to show an example fr every $n$. Can the bound depend on $W$ instead of $n$?}
Consider an instance with $n$ agents with weights $w_{1} = 1$ and $w_{i} > 1$ for $i\ge 2$.
There is one item $o$ with valuations $v_1(o) = V$ and $v_i(o) = V-\epsilon = V_{2nd}$ for $i\ge 2$, and the other items are worth $0$ to all agents.

By \Cref{lem:single-item},
the only WEF-able allocation is to give $o$ to agent $1$.
In this case, the minimum subsidy is $s_1 = 0$ and $s_i = w_i   \frac{v_i(o)}{w_1}= w_i  V_{2nd}$.
Summing all subsidies leads to $\sum_{i\geq 2} w_i  V_{2nd} = (W-1)  V_{2nd}$.
As $\epsilon$ can be arbitrarily small, we get a lower bound of $(W-1)  V$.  
% \erel{TODO: change to Vmax?
% MAYBE: use $V_{2nd}$?
% (maybe change to $V_{2nd}$)?
% } \noga{done}
\end{proof}

In \Cref{example_intro}, we showed that the iterated-maximum-matching algorithm \cite{brustle2020one} might produce an allocation that is not WEF-able.

We now introduce a new algorithm, \Cref{alg:general-additive}, which extends the iterated-maximum-matching approach to the weighted setting, assuming all weights are integers. The algorithm finds a one-to-many maximum matching between agents and items, ensuring that each agent $i\in N$ receives exactly $w_i$ items.
If the number of items remaining in a round is less than $W$, we add dummy items (valued at $0$ by all agents) so that the total number of items becomes $W$.

In \Cref{example_intro}, we add 9 dummy items, 
and perform a one-to-many maximum-value matching between agent and items, resulting in a WEF-able allocation: $A_1 = \emptyset, A_2 = \{o_1, o_2\}$.

The algorithm runs in $\lceil m/W \rceil$ rounds.
In each round $t$, the algorithm computes a one-to-many maximum-value matching $\{ A_i^t \}_{i\in N}$ between all agents and unallocated items $O_t$, where each agent $i\in N$ receives exactly $w_i$ items. 

To achieve this, we reduce the problem to the \emph{minimum-cost network flow problem} (\citet{goldberg1989network}) by constructing a flow network and computing the maximum integral flow of minimum cost. The flow network is defined as follows:
\begin{itemize}
    \item \textbf{Layer 1 (Source Node).} a single source node $s$.
    \item \textbf{Layer 2 (Agents).} a node for each agent $i\in N$, with an arc from $s$ to $i$, having cost $0$ and capacity $w_i$.
    \item \textbf{Layer 3 (Unallocated Items).} a node for each unallocated item $o \in O_t$, with an arc from each agent $i\in N$ to item $o$, having cost $-v_i(o)$ and capacity $1$.
    \item \textbf{Layer 4 (Sink Node).} a single sink node $t$, with an arc from each item $o\in O_t$ to $t$, having 0 cost and capacity $1$.
\end{itemize}
Any integral maximum flow in this network corresponds to a valid matching where each agent $i\in N$ receives exactly $w_i$ items from $O_t$, and each item is assigned to exactly one agent, the result is a minimum-cost one-to-many matching based on the costs in the constructed network. Because we negate the original costs in our construction, the obtained matching $\{A_i^t\}_{i\in N}$ maximizes the total value with respect to the original valuations.
After at most $\lceil m/W \rceil$ valuations, all items are allocated.

\begin{algorithm}[t]
\caption{
\label{alg:general-additive}
Weighted Sequence Protocol For Additive Valuations and Integer weights}
\begin{algorithmic}[1]
\Require Instance $(N,M,v, \mathbf{w})$ with additive valuations.
\Ensure WEF-able allocation $A$ with total required subsidy of at most $(W-w_1)  V$.
\State $A_{i} \gets \emptyset, \forall i\in N$
\State $t \gets 1$; $O_1 \gets M$
\While{$O_t \neq \emptyset$}
        Construct the flow network $G'=(V',E')$:
        \begin{itemize}
            \item define $V'= N \cup O_t \cup \{s,t\}$.
            \item Add arcs with the following properties:
            \begin{itemize}
                \item From $s$ to each agent $i\in N$ with cost $0$ and capacity $w_i$.
                \item From each agent $i\in N$ to each unallocated item $o \in O_t$, with cost $-v_i(o)$ and capacity 1.
                \item From each unallocated item $o\in O_t$ to $t$ with cost $0$ and capacity $1$.
            \end{itemize}
        \end{itemize}
        Compute an \emph{integral maximum flow of minimum cost} on $G'$, resulting in the one-to-many matching $\{A_i^t\}_{i\in N}$\;
        \State Set $O_{t+1} \gets O_t \backslash \cup_{i\in N}A_i^t$
        \State $t \gets t + 1$
        \EndWhile
      \State return $A$
\end{algorithmic}
\end{algorithm}

\begin{proposition} \label{general additive: WEF-able}
    %The resulting allocation by \Cref{alg:general-additive} is WEF-able.
     %More specifically, 
     For each round $t$ in  \Cref{alg:general-additive}, $A^t$ is WEF-able.
\end{proposition}
\begin{proof}
    We prove that, in every round $t$, the total cost added to any directed cycle in the weighted-envy graph is non-positive. Combined with \Cref{thm:wefable--iff-no-cycles}, this shows that $A^t$ is WEF-able for every round $t \in T$.
    
    Let $A^t$ the allocation computed by \Cref{alg:general-additive} at iteration $t$. Note that \Cref{alg:general-additive} is deterministic. 
    Let $C$ be any directed cycle in $G_{A^t,w}$, and denote $C = (i_1,..., i_r)$. To simplify notation, we consider $i_1$ as $i_{r+1}$. 

    Given the allocation $A^t$ and the cycle $C$, we construct a random alternative allocation $B^t$ as follows: 
    for each agent $i_j \in C$, we choose one item $o_{i_{j+1}}^t$ uniformly from $i_{j+1}$'s bundle and transfer it to $i_{j}$'s bundle \footnote{Recall that at each iteration, each agent $i_j$ receives exactly $w_{i_j}$ items.}.
    The expected value of $v_{i_j}(o_{i_{j+1}}^t)$, the value of the item removed from $i_j$'s bundle, can be computed as the average value of all items in $A^t_{i_j}$: $$\frac{\sum_{o\in A^t_{i_j}} v_{i_j}(o)}{w_{i_j}} = \frac{v_{i_j}(A^t_{i_j})}{w_{i_j}}.$$ Similarly, the expected value of $v_{i_j}(o_{i_j}^t)$, the value of the item added the $i_j$'s bundle, is $\frac{v_{i_j}(A^t_{i_{j+1}})}{w_{i_{j+1}}}$.
    Thus, the expected change in value between $B^t$ and $A^t$ is
    \begin{align*}
        & \mathbb{E}\left[\sum_{i\in N} \left( v_i(B_i^t) - v_i(A^t_i) \right)\right] = 
        \sum_{i_j \in C} \left( \mathbb{E}\left[v_{i_j}(o_{i_{j+1}}^t)\right] - \mathbb{E}\left[v_{i_j}(o_{i_j}^t)\right] \right) = \\
        &
        \sum_{i_j\in C} \frac{v_{i_j}(A^t_{i_{j+1}})}{w_{i_{j+1}}} - \frac{v_{i_j}(A^t_{i_j})}{w_{i_j}}
    .
    \end{align*} 
    This is exactly the total cost of cycle $C$.
    % which by assumption is greater than 0.
    According to \Cref{alg:general-additive}, $A^t$ maximizes the total value among all allocations in which each agent $i$ receives exactly $w_i$ items. Therefore, the left-hand side of the above expression, which is the difference between the sum of values in $B^t$ and the sum of values in $A^t$, must be at most $0$. But the right-hand side of the same expression is exactly the total cost of $C$. Therefore,
    $$ 0 \geq \mathbb{E}\left[\sum_{i\in N} \left( v_i(B^{t}_i) - v_i(A^{t}_i) \right)\right] = cost_{A^{t}}(C),$$
    so the cost of every directed cycle is at most $0$, as required.
\end{proof}
As the allocation in each iteration is WEF-able, the output allocation $A$ is WEF-able too.
To compute an upper bound on the subsidy, we adapt the proof technique in \cite{brustle2020one}. Let $A^t$ be the output allocation from \Cref{alg:general-additive}, computed in iteration $t$.
For each $i\in N$, we define the modified valuation function as follow:
\[
\Bar{v}_i(A_j^t) = \begin{cases}
    \frac{v_i(A_i^t)}{w_i} & j = i \\
    \frac{v_i(A_j^T)}{w_j} & j\neq i, t = T \\
    \max{\left(\frac{v_i(A_j^t)}{w_j}, \frac{v_i(A_i^{t+1})}{w_i}\right)} & j\neq i, t < T
\end{cases}
\]
Under the modified valuations, for any two agents $i,j \in N$, the modified-cost assigned to the edge $(i,j)$ in the envy graph (with unit weights) is defined as $\overline{cost}_{A}(i,j) = \Bar{v}_i(A_i) - \Bar{v}_i(A_j)$. Moreover, the modified-cost of a path $(i_1,...,i_k)$ is $\overline{cost}_{A}(i_1,...,i_k) = \sum_{j=1}^{k-1} \overline{cost}_{A}(i_{j}, i_{j+1})$.

\begin{lemma} \label{lemm:minimum cost subsidy}
    Let $A$ be a WEF-able allocation. 
    For any positive number $z$,
    if $cost_A(i,k) \geq -z$ for every edge $(i,k)$ in $G_{A,w}$, then the maximum subsidy required is at most $w_i z$ per agent $i \in N$.
\end{lemma}
\begin{proof} 
    Assume $P_i(A) = (i \ldots j)$ is the highest-cost path from $i$ in $G_{A,w}$. Note that $\ell_i(A) = cost_A(P_i(A))$.
    
    Then it holds for the cycle $C = (i \ldots j)$ that \[cost_A(C) = \ell_i(A) + cost_A(j,i) .\]
    
    By \Cref{thm:wefable--iff-no-cycles}, $cost_A(C) \leq 0$, thus, \[\ell_i(A) \leq -cost_A(j,i) \leq z.\]

    Therefore, $s_i = w_i   \ell_i(A) \leq w_i z$.
\end{proof}

% In most allocations, including the one resulting from \Cref{alg:general-additive}, it is not always true that $cost_A(i,k) \geq -V$ for every edge $(i,k)$. We apply the technique from Brustle et al. ~\shortcite{brustle2020one} to 
We use \Cref{lemm:minimum cost subsidy} with a modified valuation function.% $\Bar{v}_i$, derived from the weighted valuation $\frac{v_i(A_i)}{w_i}$.

We prove that an allocation that is WEF-able for the original valuations is also WEF-able for the modified valuations (\Cref{modified function wef-able}), and that the maximum subsidy required by each agent for the original valuations is bounded by the subsidy required for the modified valuations (\Cref{max sub modified is max sub original}).

Next, we demonstrate that under the modified valuations, the cost of each edge is at least $-V$. Finally, by \Cref{lemm:minimum cost subsidy}, we conclude that the maximum subsidy required for any agent $i\in N$ is $w_i   V$ for the modified valuation $\Bar{v}$ (\Cref{max subsidy for modified valuations}) and for the original valuations $v$ as well.

 \begin{observation} \label{observation:modified func}
For agent $i \in N$ and round $t\in [T]$ it holds that:
    \begin{enumerate}
        \item $
        \Bar{v}_i(A_i^t)
        =
        \frac{v_i(A_i^t)}{w_i}$.
        \item For agent $j\neq i \in N$, $
        \Bar{v}_i(A_j^t)\geq \frac{v_i(A_j^t)}{w_j}$;
        hence $\Bar{v}_i(A_j^t) - \Bar{v}_i(A_i^t)\geq \frac{v_i(A_j^t)}{w_j}-\frac{v_i(A_i^t)}{w_i}$.
    \end{enumerate}
\end{observation}

\begin{proposition} \label{modified function wef-able}
    Assume $A$ is WEF-able under the original valuations $v$. Then, $A$ is EF-able (i.e., WEF-able with unit weights) under the modified valuations $\Bar{v}$. 
\end{proposition}

\begin{proof} 
By \Cref{thm:wefable--iff-no-cycles}, it is sufficient to prove that all directed cycles in the envy graph (with the modified valuations and unit weights) 
have non-positive total cost.
We prove a stronger claim: in every round $t$, 
the total modified-cost added to every directed cycle $C$ is non-positive. 
 Let $A^t$ be the allocation computed by \Cref{alg:general-additive} at iteration $t$. 
Suppose, contrary to our assumption, that there exists a cycle $C=(i_1,...,i_r)$ and a round $t$
in which the modified-cost added to $C$ is positive. To simplify notation, we consider $i_1$ as $i_{r+1}$. 
This implies that
\begin{align} \label{contradiction: modified valuation}
    \sum_{j=1}^r \Bar{v}_{i_j}(A_{i_{j+1}}^t) > \sum_{j=1}^r \Bar{v}_{i_j}(A_{i_j}^t).
\end{align}
There are several cases to consider.

\underline{\bf Case 1:}  All arcs $i_j\to i_{j+1}$ in $C$ have $$\Bar{v}_{i_j}(A_{i_{j+1}}^t) = \frac{v_{i_j}(A_{i_{j+1}}^{t})}{w_{i_{j+1}}}$$
    (in particular, this holds for $t = T$).
    In this case, inequality \eqref{contradiction: modified valuation} implies 
    $$\overline{cost}_{A^t}(C) = \sum_{j=1}^r \frac{v_{i_j}(A_{i_j}^{t})}{w_{i_{j+1}}} - \frac{v_{i_j}(A_{i_j}^{t})}{w_{i_j}} > 0.$$ Combined with \Cref{thm:wefable--iff-no-cycles}, this contradicts \Cref{general additive: WEF-able}, which states that $A^t$ is WEF-able.

\underline{\bf Case 2:}    
    All arcs $i_j\to i_{j+1}$ in $C$ have $$\Bar{v}_{i_j}(A_{i_{j+1}}^t) = \frac{v_{i_j}(A_{i_j}^{t+1})}{w_{i_j}}.$$
    In this case, inequality 
    \eqref{contradiction: modified valuation} implies 
    $$\sum_{j=1}^r \frac{v_{i_j}(A_{i_j}^{t+1})}{w_{i_j}} > \sum_{j=1}^r \frac{v_{i_j}(A_{i_j}^{t})}{w_{i_j}}.$$ 
    Notice that all the items in $A_{j_{1}}^{t+1},..., A_{j_{r}}^{t+1}$ are available at iteration $t$, which contradicts the optimality of $\{A_   i^t\}$.
    
\underline{\bf Case 3:} 
Some arcs $i_j\to i_{j+1}$ in $C$ satisfy Case 1 and the other arcs satisfy Case 2.
Let $l \geq 1$ be the number of arcs in $C$ that satisfy Case 2.
We decompose $C$ into a sequence of $l$ edge-disjoint paths, denoted $P_1, ... , P_l$, such that the last node of each path is the first node of the next path,
and in each path, only the last edge satisfies Case 2.
% where the paths are not disjoint at the nodes but are connected to each other at a single node. 
Formally, suppose that some path contains $k\geq 1$ agents, denoted as $i_1, ..., i_{k}$, and $k-1$ arcs.
Then  for each $1\leq j \leq k - 2$, $\Bar{v}_{i_{j}}\left(A_{i_{j+1}}^t\right) = \frac{v_{i_{j}}\left(A_{i_{j+1}}^t\right)}{w_{i_{j+1}}}$ and $\Bar{v}_{i_{k-1}}\left(A_{i_{k}}^t\right) = \frac{v_{i_{k-1}}\left(A_{i_{k-1}}^{t+1}\right)}{w_{i_{k-1}}}$. 
    %\erel{Where is $i_k$?}
    Since $\overline{cost}_{A^t}(C) > 0$, there exits a path $P = (i_1,..., i_k)$ where $\overline{cost}_{A^t}(P) > 0$, which implies that: \\
    \resizebox{\textwidth}{!}{$
        0 < \sum_{j = 1}^{k-1} \left( \Bar{v}_{i_{j}}(A_{i_{j+1}}^t) - \Bar{v}_{i_{j}}(A_{i_{j}}^t) \right) =
    \sum_{j=1}^{k-2} \left(\frac{v_{i_{j}}(A_{i_{j+1}}^t)}{w_{i_{j+1}}} - \frac{v_{i_{j}}(A_{i_{j}}^t)}{w_{i_{j}}} \right) + \frac{v_{i_{k-1}}(A_{i_{k-1}}^{t+1})}{w_{i_{k-1}}} - \frac{v_{i_{k-1}}(A_{i_{k-1}}^{t})}{w_{i_{k-1}}}$.}\\
    The rest of the proof is similar to the proof of \Cref{general additive: WEF-able}.
    We construct another allocation $B^t$ randomly as follows: 
    \begin{enumerate}
        \item For each agent $1\leq j \leq k-1$, we choose one item $o_{i_{j+1}}^t$ uniformly from $i_{j+1}$'s bundle and transfer it to $i_{j}$'s bundle.
        \item We choose one item $o_{i_{1}}^t$ uniformly from $i_1$'s bundle to remove.
        \item We choose one item $o_{i_{k-1}}^{t+1}$ uniformly from $A_{i_{k-1}}^{t+1}$ and add it to $i_{k-1}$'s bundle. 
    \end{enumerate}
    Thus, the expected change in value between $B^t$ and $A^t$ is
    \begin{align*}
        &\mathbb{E}\left[\sum_{i\in N} \left( v_i(B_i^t) - v_i(A^t_i) \right)\right] = 
         \sum_{1\leq j\leq k-2} \left( \mathbb{E}\left[v_{i_j}(o_{i_{j+1}}^t)\right] - \mathbb{E}\left[v_{i_j}(o_{i_j}^t)\right] \right) + \\
         &\mathbb{E}\left[v_{i_{k-1}}(o_{i_{k-1}}^{t+1})\right] - \mathbb{E}\left[v_{i_{k-1}}(o_{i_{k-1}}^t)\right] = 
        \sum_{1 \leq j \leq k-2} \frac{v_{i_j}(A^t_{i_{j+1}})}{w_{i_{j+1}}} - \frac{v_{i_j}(A^t_{i_j})}{w_{i_j}} + \\
        &\frac{v_{i_{k-1}}(A_{i_{k-1}}^{t+1})}{w_{i_{k-1}}} - \frac{v_{i_{k-1}}(A_{i_{k-1}}^{t})}{w_{i_{k-1}}}.
    \end{align*}
    This is exactly the cost of $P$ which by assumption is greater than 0.
    However, according to \Cref{alg:general-additive}, $A^t$ maximizes the value of an allocation where each agent $i$ receives $w_i$ items among the set of $O^t$ items. Therefore,
     $$0 \leq \mathbb{E}\left[\sum_{i\in N} \left( v_i(B^{t}_i) - v_i(A^{t}_i) \right)\right] = cost_{A^{t}}(P)$$ 
    leading to a contradiction.

To sum up, $A^t$ is WEF-able under the original valuations $v$ (with weights $w$), and under the modified valuations $\Bar{v}$ (with unit weights).
\end{proof}
\begin{proposition} \label{max sub modified is max sub original}
    For the allocation $A$ computed by \Cref{alg:general-additive}, the subsidy required by an agent given $v$ (with weights $w$) is at most the subsidy required given $\Bar{v}$ (with unit weights).
\end{proposition}
\begin{proof}
    Given \Cref{observation:modified func}, for each $i,j \in N$,
    $$\Bar{v}_i(A_j) - \Bar{v}_i(A_i) \geq \frac{v_i(A_j)}{w_j} - \frac{v_i(A_i)}{w_i}.$$
    Thus, the cost of any path in the envy graph under the modified function and unit weights 
    is at least the cost of the same path in the weighted envy-graph with the original valuations. 
\end{proof}
\begin{proposition} \label{max subsidy for modified valuations}
    For the allocation $A$ computed by \Cref{alg:general-additive}, the subsidy to each agent is at most $w_i   V$ for the modified valuation profile $\Bar{v}$.
\end{proposition}
\begin{proof}
    By \Cref{modified function wef-able}, the allocation $A$ is WEF-able under the valuations $\Bar{v}$. 
    Together with \Cref{lemm:minimum cost subsidy}, if for each $i,j\in N$ it holds that $\Bar{v}_i(A_j) - \Bar{v}_i(A_i) \geq -V$, the subsidy required for agent $i\in N$ is at most $w_i   V$ for $\Bar{v}$.

    \begin{align*}
        &\Bar{v}_i(A_j) - \Bar{v}_i(A_i) = \sum_{t\in [T]}\Bar{v}_i(A_j^t) - \sum_{t\in [T]}\Bar{v}_i(A_i^t) =  \\ &\sum_{t\in [T-1]}\max\Big\{\frac{v_i(A_j^t)}{w_j}, \frac{v_i(A_i^{t+1})}{w_i}\Big\} + \frac{v_i(A_j^T)}{w_j}- \sum_{t\in [T]}\frac{v_i(A_i^t)}{w_i} \geq \\ &
         \sum_{t\in [T-1]}\frac{v_i(A_i^{t+1})}{w_i}+ \frac{v_i(A_j^T)}{w_j}- \sum_{t\in [T]}\frac{v_i(A_i^t)}{w_i} = \\ &\frac{v_i(A_j^T)}{w_j} - \frac{v_i(A_i^1)}{w_i} \geq - \frac{v_i(A_i^1)}{w_i}.
    \end{align*}
Since $A_i^1$ contains exactly $w_i$ items, 
$-v_i(A_i^1) \geq - w_i  V$.
Hence, $\Bar{v}_i(A_j) - \Bar{v}_i(A_i) \geq -\frac{w_i   V}{w_i} = -V$.
\end{proof}

We are now prepared to prove the main theorem.
\begin{theorem} \label{theorem: sub general additive}
    For additive valuations and integer entitlements, \Cref{alg:general-additive} computes in polynomial time a WEF-able allocation, where the subsidy to each agent is at most $w_i V$ and the total subsidy is at most $(W-w_1)V$.
\end{theorem}
\begin{proof}
By \Cref{general additive: WEF-able}, $A$ is WEF-able under the original valuations. Combined with \Cref{modified function wef-able} and \Cref{max subsidy for modified valuations} , $A$ is also WEF-able under the modified valuations and requires a subsidy of at most $w_i V$ for each agent $i\in N$. 

\Cref{max sub modified is max sub original}, implies that under the original valuations, the required subsidy for each agent $i\in N$ is at most $w_i V$. 
By \Cref{max_path_subsidy}, there is at least one agent who requires no subsidy, so the required total subsidy is at most $(W-w_1)  V$.

For the runtime analysis, the most computationally intensive step in \Cref{alg:general-additive} is solving the maximum integral flow of minimum cost in $G'$. The flow network $G'$ consists of at most $n+m+2$ nodes and at most $n+m+mn$ arcs. 
By \citet{goldberg1989finding}, this can be done in time polynomial in $n,m$:
\begin{multline*}
         O\left(\left(n+m+2\right)\left( n+m+mn \right) \log\left(n+m+2\right) \right. \\
         \left.\min\{ \log \left(\left(n+m+2\right) V \right), \left( n+m+mn \right) \log \left(n+m+2\right) \}\right).
\end{multline*}
\end{proof}

The WEF condition is invariant to multiplying the weight vector by a scalar.
This can be used in two ways:

(1) If the weights are not integers, but their ratios are integers, we can still use Algorithm~\ref{alg:general-additive}.
 For example, if $w_1=1/3$ and $w_2=2/3$ (or even if $w_i$'s are irrational numbers such as $w_1=\sqrt{2}$ and $w_2=2\sqrt{2}$), Algorithm 1 works correctly by resetting $w_1=1$ and $w_2=2$.

(2) If the weights are integers with greatest common divisor (gcd) larger than 1, we can divide all weights by the gcd to get a better subsidy bound:

% Dividing all weights by a common factor does not affect the envy in any way. Therefore, we can divide each weight by $\gcdw$ and get the improved upper bound $\frac{W-\wmin}{\gcdw} V$.

\begin{lemma} \label{cor: sub general additive}
    For additive valuations and integer entitlements, there exists an algorithm that computes in polynomial time a WEF-able allocation where the subsidy to each agent is at most $w_i V/\gcdw$ and the total subsidy is at most $(W-w_1)V/\gcdw$, where $\gcdw$ is the greatest common divisor of all the $w_i$.
\end{lemma}
\begin{proof}
Algorithm~\ref{alg:general-additive} works correctly, even if we divide each $w_i$ by the greatest common divisor of $w_i$'s.
In other words, letting $d={\rm gcd}(w_1,...,w_n)$, $w'_i=w_i/d$, $W'=W/d$, and running Algorithm 1 with $w'_i$'s, we get the bound $(W'-w'_{\min})V$ of the total subsidy.
\end{proof}

A discussion about the tightness of the bound can be found in Appendix \ref{alg:general-additive tightness}.

\section{WEF Solutions for Additive Identical Valuations}
\label{sec:identical-additive}
In this section
%our focus shifts towards refining algorithms tailored for identical additive valuations, a foundational case within valuation scenarios. 
%we assume that all agents have the same additive valuation functions, denoted by $v$.
This section deals with the case where all agents have identical valuations, that is, $v_i\equiv v$ for all $i\in N$.

We present a polynomial-time algorithm for finding a WEF-able allocation with a subsidy bounded by $V$ per agent and a total subsidy bounded by $(n-1)V$.
The following example shows that this bound is tight for any weight vector:
\begin{example}
    \label{identical-additive-tightness}
    Consider $n$ agents with integer weights $w_1 \leq \cdots \leq w_n$ and $1 + \sum_{i\in N} \left( w_i - 1 \right)$ items all valued at $V$.
    
    To avoid envy, each agent $i$ should receive a total utility of $w_i V$, so the sum of all agents' utilities would be $W V$. 
    
    As the sum of all values is $(W-(n-1))V$, a total subsidy of at least $(n-1)V$ is required
    (to minimize the subsidy per agent, each agent $i\in N$ should receive $w_i - 1$ items, except for the agent with the highest entitlement (agent $n$), who should receive $W$ items.
    
    The value per unit entitlement of each agent $i<n$ is $V(w_i-1)/w_i$, 
    and for agent $n$ it is $V$.
    Therefore, to avoid envy, each agent $i<n$ should receive a subsidy of $w_i \left(1 - \frac{w_i - 1}{w_i}\right)V = V$ and the total subsidy required is $(n-1)V.$
\end{example}
%. The established subsidy bound, marked at $\max*_{o\in M}(v(o)))$ for each agent $i \in N$, is significant in ensuring fair outcomes.

%We start by observing that, with identical valuations, the cost of any path in the weighted envy graph is determined only by the agents at the endpoints of that path.
We start by observing that, with identical valuations, the cost of any path in the weighted envy graph is determined only by the agents at the endpoints of that path.
\begin{observation}
\label{cost_of_path_identical_valuatoins}
Given an instance with identical valuations,
let $A$ be any allocation, and denote by $P$ any path in the weighted envy-graph of $A$ between agents $i,j\in N$. Then,
\begin{align*}
    cost_A(P) = \frac{v(A_j)}{w_j}-\frac{v(A_i)}{w_i}.
\end{align*}
This is because the path cost is 
$$
    \sum_{(h,k)\in P} cost_A(h,k)= \sum_{(h,k)\in P}\frac{v(A_k)}{w_k}-\frac{v(A_h)}{w_h},
$$ 
and the latter sum is a telescopic sum that reduces to the difference of its last and first  element.
\end{observation}
Using \Cref{cost_of_path_identical_valuatoins}, it is easy to prove the following Lemma.
\begin{lemma}\label{thm:wefable--iff-no-cycles1}
    With identical valuations, every allocation is WEF-able.
\end{lemma}
\begin{proof}
Consider an allocation $A$ over $m$ items and $n$ agents with identical valuations $v$ and weights $w$. 
Let $C=(i_{1}, ..., i_{r})$ be a cycle in $G_{A,w}$. 
Then,
\begin{equation}
\begin{split}
    & cost_{A}(C)
    = 
    \\
    &
    \sum_{k = 1}^{r-1} \left(\frac{v(A_{i_{k+1}})}{w_{i_{k+1}}} - \frac{v(A_{i_{k}})}{w_{i_{k}}}\right) + \frac{v(A_{i_1})}{w_{i_1}} - \frac{v(A_{i_r})}{w_{i_r}} = 
    \\ 
    &
    \sum_{k=1}^{r}\left(\frac{v(A_{i_k}}{w_{i_k}} - \frac{v(A_{i_k}}{w_{i_k}}\right) = 0.
\end{split}
\end{equation}
    Hence, by condition \ref{condition_b_theroem_1} of \Cref{thm:wefable--iff-no-cycles}, $A$ is WEF-able.
\end{proof}
Our algorithm for finding a WEF-able allocation with bounded subsidy is presented as \Cref{alg:identical-additive}.

%It first orders the items in $M$ 
%It operates on input sets, namely a set of items denoted as $M$ and a set of agents denoted as $N$. The agents possess additive identical valuations denoted by $v$, weights represented as $(w_i){i\in N}$, and a predefined ordering 
%based on their common value, such that 
%$v(o_1) \geq v(o_2) \geq ... \geq v(o_m)$. 

%The algorithm then iteratively traverses the items. At iteration $t$ it selects the agent that minimizes the expression $\frac{v(A_{i}^{t-1}\cup{\{o_{t}\}})}{w_i}$, with ties are broken in favor of the agent with the higher index, which corresponds to the agent with the larger $w_i$.

%Intuitively, this selection minimizes the likelihood that weighted envy is generated.
%Subsequently, the algorithm allocates item $o_t$ to the selected agent.
The algorithm traverses the items in an arbitrary order. At each iteration it selects the agent that minimizes the expression $\frac{v(A_i\cup{\{o\}})}{w_i}$, with ties broken in favor of the agent with the larger $w_i$, and allocates the next item to that agent.
Intuitively, this selection minimizes the likelihood that weighted envy is generated.

\begin{algorithm}
\caption{Weighted Sequence Protocol For Additive Identical Valuations}\label{alg:identical-additive}
\begin{algorithmic}[1]
\Require Instance $(N,M,v, \mathbf{w})$ with additive identical valuations.
\Ensure WEF-able allocation $A$ with total required subsidy of at most $(n-1)  V$.
\State $A_{i} \gets \emptyset, \forall i\in N$
\For{$o: 1$ to $m$}
\State $I^t \gets \arg\min_{i\in N}\frac{v\left(A_i \cup \{o\}\right)}{w_i}$
\State $i^t \gets \max_{i\in I^t}\left(i\right)$
\State Add $o$ to $A_{i^t}$
\EndFor
\State return $A$
\end{algorithmic}
\end{algorithm}

The following example illustrates \Cref{alg:identical-additive}:
 \begin{example}\label{example:additive identical}
    Consider two agents, denoted as $i_1$ and $i_2$, with corresponding weights $w_1 = 1$ and $w_2 = \frac{7}{2}$, and three items, namely $o_1, o_2, o_3$, with valuations $v(o_1) = v(o_2) = v(o_3) = 1$, \Cref{alg:identical-additive} is executed as follows:
    \begin{enumerate}
    \item for $t=1$, the algorithm compares $\frac{v(o_{1})}{w_1} = 1 $ and $\frac{v(o_{1})}{w_2} = \frac{2}{7}$. Consequently, the algorithm allocates item $o_1$ to agent $i_2$, resulting in $A_1^1 = \emptyset$ and $A_2^1 = \{o_1\}$.
    \item for $t=2$, the algorithm compares $\frac{v(o_{2})}{w_1} = 1$ and $ \frac{v(A_2^1 \cup \{o_{2}\})}{w_2} = \frac{2}{\frac{7}{2}} = \frac{4}{7}$. Subsequently, the algorithm allocates item $o_2$ to agent $i_2$, resulting in $A_1^2 = \emptyset$ and $A_2^2 = \{o_1, o_2\}$.
    \item for $t=3$, the algorithm compares $\frac{v(o_3)}{w_1} = 1$ and $\frac{v(A_2^2 \cup \{o_3\})}{w_2} = \frac{3}{\frac{7}{2}} = \frac{6}{7}$, Consequently, item $o_3$ is allocated to agent $i_2$, resulting in $A_1^3 = \emptyset$ and $A_2^3 = \{o_1, o_2, o_3\}$.
    \item agent $i_1$ envies agent $i_2$ by an amount of $\frac{v(A_2^3)}{w_2}-\frac{v(A_1^3)}{w_1} = \frac{3}{\frac{7}{2}} = \frac{6}{7}$, and conversely, agent $i_2$ envies agent $i_1$ by $\frac{v(A_1^3)}{w_1}-\frac{v(A_2^3)}{w_2} = -\frac{6}{7}$
    \item In order to mitigate envy, $s_{1} = \frac{6}{7}$ and $s_{2} = 0$.
    \end{enumerate}
    \end{example}
    
Example \ref{example:additive identical} illustrates that the resulting allocation may not be\emph{WEF(1,0)} ---
%Even if one item removed from $A_2$'s bundle, the envy of $i_1$ for $i_2$ persists. 
\Cref{alg:identical-additive} might allocate all items to the agent with the highest entitlement.
However, the outcome is always WEF$(0,1)$: %\rica{It is the same reason that the general additive valuation case is not WEF(1,0). Algorithm 1 may allocate all items to one agent. We should point out this reason briefly here.}\noga{OK, added} 
%However, it does satisfy the \emph{WEF(0,1)} criterion. The following Proposition holds based on the selection rule, which ensures that the cost for any agent $i\in N$ towards the agent who received the new item is bounded by $\frac{V}{w_i}$.

\begin{proposition} \label{identical additive wef01}
    For additive identical valuations, \Cref{alg:identical-additive} computes a $WEF(0,1)$ allocation.
\end{proposition}
\begin{proof}
We prove by induction that at each iteration, $A$, the resulting allocation from \Cref{alg:identical-additive}, satisfies WEF$(0,1)$. 
The claim is straightforward for the first iteration. Assume the claim holds for the $(t-1)$-th iteration, and prove it for the $t$-th iteration. Let $o$ be the item assigned in this iteration and $i^t$ be the agent receiving this item.
Agent $i^t$ satisfies WEF$(0,1)$ due to the induction hypothesis.
For $j \neq i^t$, by the selection rule, $\frac{v(A_{i^t})}{w_{i^t}} \leq \frac{v(A_j \cup \{o\})}{w_j}$.
This is exactly the definition of WEF$(0,1)$.
\end{proof}
Based on \Cref{cost_of_path_identical_valuatoins} and \Cref{identical additive wef01}, we show that each path starting at agent $i\in N$ under $WEF(0,1)$ allocation is bounded by $\frac{V}{w_i}$. 
\begin{proposition}
\label{prop:wef01 upper bound}
With identical additive valuations, for every WEF$(0,1)$ allocation $A$,
$\ell_i(A) \leq \frac{V}{w_i}$, for all $i\in N$.
\end{proposition}
\begin{proof} 
For each agent $i\in N$, 
denote the highest-cost path starting at $i$ in that graph by $P_{i}(A) = (i, ..., j)$ for some agent $j\in N$. 
Then by \Cref{cost_of_path_identical_valuatoins},
$\ell_i(A) = cost_{A}(P_i(A)) = \frac{v(A_{j})}{w_{j}} - \frac{v(A_{i})}{w_i}$.

From the definition of WEF$(0,1)$, \Cref{prop:wef01 upper bound} implies that this difference is at most $\frac{v(o)}{w_i}$ for some object $o\in A_j$. 
Therefore, the difference is at most $\frac{V}{w_i}$.
\end{proof}
\begin{theorem}\label{thm:wefable--iff-no-cycles3}
    For additive identical valuations, \Cref{alg:identical-additive} computes a WEF-able allocation $A$ in $O(mn)$ time, where $\forall i \in N$: $s_{i} \leq V$. 
    Therefore, the total subsidy required is at most $(n-1)  V$.
\end{theorem}
\begin{proof}
    Let $A$ be the allocation output by \Cref{alg:identical-additive} under additive identical valuations after $m$ iterations. 
    From \Cref{thm:wefable--iff-no-cycles1}, it can be deduced that $A$ is WEF-able. Together
    \Cref{identical additive wef01} and
    \Cref{prop:wef01 upper bound} imply that, to achieve weighted-envy-freeness under identical additive valuations for the allocation computed by \Cref{alg:identical-additive}, the required subsidy per agent $j \in N$ is at most $ w_j   \frac{V}{w_j} = V$. 
    In combination with \Cref{max_path_subsidy}, the total required subsidy is at most $(n-1)  V$. 
    Note that $W \geq n   w_{1}$. Therefore, this bound is better than the one proved in 
    \Cref{theorem: sub general additive}: $\left(W - w_i \right)   V \geq \left(n - 1 \right) w_i  V \geq \left(n - 1 \right)   V $.
    
    We now analyze \Cref{alg:identical-additive}'s time complexity.
    
    The loop in the algorithm runs for $m$ times. where at each iteration $t\in[m]$, finding the set of agents $I^{t}$ and the agent $i^t$ within it takes $O(n)$. 
    Also, allocating the item $o$ to the agent $i^t$ takes $O(1)$.
    
    To sum up, \Cref{alg:identical-additive} runs in $O(mn)$.
\end{proof}

%The upper bound of $(n-1)V$ is tight even for equal entitlements \cite{halpern2019fair}.
%Interestingly, when either the valuations or the entitlements are identical, the worst-case upper bound depends on $n$, whereas when both valuations and entitlements are different, the bound depends on $W$.
Note that $W \geq n   w_1$. Therefore, this bound is better than the one proved in 
    \Cref{theorem: sub general additive} for integer weights: $\left(W - w_i \right)   V \geq \left(n - 1 \right) w_i  V \geq \left(n - 1 \right)   V $.

The upper bound of $(n-1)V$ is tight even for equal entitlements (\citet{halpern2019fair}).
Interestingly, when either the valuations or the entitlements are identical, the worst-case upper bound depends on $n$, whereas when both valuations and entitlements are different, the bound depends on $W$.

\section{WEF Solutions for binary additive Valuations} 
\label{sec:binary-additive}
In this section we focus on the special case of agents with binary additive valuations.
We start with a lower bound on the subsidy.

\begin{proposition}
\label{prop:lower-bound-binary}
    For every $n \geq 2$ and weight vector $\mathbf{w}$, there is an instance with $n$ agents with binary valuations in which the required subsidy in any WEF allocation is at least
    $\frac{W}{w_2} - 1$.
\end{proposition}
\begin{proof}
    Agents $1$ and $2$ value the item at $1$ and the others at $0$.
If agent $i \in \{1,2\}$ gets the item,
then the other agent $j\neq i \in \{1,2\}$ must get subsidy $\frac{w_j}{w_i}$.
To ensure that other agents do not envy $j$'s subsidy, every other agent $k \not\in \{1,2\}$ 
must get subsidy $\frac{w_k}{w_i}$.
The total subsidy is $\frac{W}{w_i} - 1$.
The subsidy is minimized by giving the item to agent $2$, since $w_2\geq w_1$.
This gives a lower bound of $\frac{W}{w_2} - 1$.
\end{proof}

%We now show how to compute a WEF-able allocation where the subsidy given to each agent $i\in N$ is at most $\frac{w_i}{w_1} V = \frac{w_i}{w_1}$.
Below, we show how to compute a WEF-able allocation where the subsidy given to each agent $i\in N$ is at most $\frac{w_i}{w_1} V = \frac{w_i}{w_1}$, and the total subsidy is at most $\frac{W}{w_1} - 1$.

In the case of binary valuations, \Cref{alg:general-additive} is inefficient in three ways: 
\begin{enumerate}
    \item The maximum-cost matching does not always prioritize agents with higher entitlements.
    \item There may be situations where an agent prefers items already allocated in previous iterations, while the agent holding those items could instead take unallocated ones.
    \item The algorithm works only for agents with integer weights.
\end{enumerate}

We address these issues by adapting the \emph{General Yankee Swap (GYS)} algorithm introduced by Viswanathan et al. in \cite{viswanathan2023general}.

GYS starts with an empty allocation for all agents.
We add a dummy agent $i_0$ and assume that all items are initially assigned to $i_0$:
$A_{i_0} = M$.

\Cref{alg:binary-additive} presents our approach for finding a WEF-able allocation with a bounded subsidy.
The algorithm runs in $T$ iterations. 
We denote by $A^t$ the allocation at the end of iteration $t$.
Throughout this algorithm, we 
divide the agents into two sets:
\begin{enumerate}
    \item $R$: The agents remaining in the game at the beginning of the iteration $t$.
    \item $N \setminus R$: The agents who were removed from the game in earlier iteration $t' < t$. Agents are removed from the game when the algorithm deduces that their utility cannot be improved.
\end{enumerate}
As long as not all the objects have been allocated, at every iteration $t\in[T]$, the algorithm looks for the agents maximizing the \emph{gain function} (\citet{viswanathan2023general}) among $R$, i.e., the agents remaining in the game at this iteration.

%In our setting we use the gain function: $\left(\frac{w_i}{v_i(A_i^t) + 1}, 1\right)$, 
%which means 
%that we first select the agents with the minimal potential to increase envy towards them (the maximum $\frac{w_i}{v_i(A_i^t) + 1}$), and then break ties in favor of the agent with the higher index, which corresponds to the agent with the larger $w_i$.
We use the gain function: 
$\frac{w_i}{v_i(A_i^{t-1}) + 1}$, 
which selects agents with the minimal potential for increasing envy.
% , determined by maximizing $\frac{w_i}{v_i(A_i^{t-1}) + 1}$. 
If multiple agents have the same value, we select one arbitrarily.

%The selected agent then faces a choice: either acquire an unallocated item or take an item allocated to another agent. Regardless of the choice, the agent's utility increases by 1. If the agent decides to take an allocated item from another agent, the affected agent must then decide whether to acquire an unallocated item or take an allocated item from a third agent to maintain their utility, and so on. This results in a path from agent $i$ to the dummy agent $i_0$, where items are passed between agents until an unallocated item is reached. 

%We call this a \emph{transfer paths}.
%When an agent is selected, the algorithm attempts to find a transfer path from this agent to $i_0$, maintaining utilities for all agents except the initiator, whose utility increases by 1. If no such path exists, the agent will no longer receive any items, and the agent is removed from the game. 

%We use the polynomial-time method by Viswanathan et al. in \cite{viswanathan2023general} to find transfer paths.

%\Cref{alg:binary-additive} differs from $GYS$ in the following way: at the beginning of iteration $t$, the algorithm first removes all agents without a transfer path originating from them (\emph{Phase 1}). Then, it selects an agent based on the gain function to allocate a new item to that agent (\emph{Phase 2}).

The selected agent then chooses either to acquire an unallocated item or take an item from another agent. In either case, their utility increases by 1. If the agent takes an item from another, the affected agent must decide whether to take an unallocated item or another allocated item to preserve their utility, and so on. This process creates a \emph{transfer path} from agent $i$ to the dummy agent $i_0$ , where items are passed until an unallocated item is reached.

Formally, we represent this as a directed graph, where nodes are agents, and an edge $(i,j)$ if and only if there exists an item in $j$'s bundle that $i$ values positively.
A \emph{transfer path} is any directed path in that graph, that ends at the dummy agent $i_0$.

When an agent is selected, the algorithm attempts to find a transfer path from that agent, preserving utilities for all agents except the initiator, whose utility increases by 1. If no path is found, the agent is removed from the game. We use the polynomial-time method by \citet{viswanathan2023general} to find transfer paths.

\Cref{alg:binary-additive} differs from $GYS$ in the following way: at the beginning of iteration $t$, the algorithm first removes all agents without a transfer path originating from them (line \ref{line:R}). Then, it selects an agent based on the gain function to allocate a new item to that agent.
For convenience, we denote by $R(t)$ the agents who have a transfer path originating from them at the beginning of iteration $t$ (line \ref{line:R}).

\begin{algorithm}[t]
\caption{
\label{alg:binary-additive}
Weighted Sequence Protocol For Additive Binary Valuations}
\begin{algorithmic}[1]
\Require Instance $(N,M,v, \mathbf{w})$ with additive binary valuations.
\Ensure WEF-able allocation $A$ with total required subsidy of at most $\frac{W}{w_1} - 1$.
\State $A_{i_{0}} \gets M$, and $A_{i}^0 \gets \emptyset$ for each $i\in N$ \Comment{All items initially are unassigned}
\State $t \gets 1$
\State $R \gets N$
\While{$R \neq \emptyset$}
    \State Remove from $R$ all agents who do not have a transfer path starting from them \label{line:R}
    \State $i^t \gets \arg \max_{i \in R} \left( \frac{w_i}{v_i(A_i^{t-1}) + 1}\right)$ \Comment{\footnotesize Choose the agent who maximizes the gain function \normalsize}
    \State Find a transfer path starting at $i^t$ 
    \Comment{\footnotesize For example, one can use the BFS algorithm to find a shortest path from $i^t$ to $i_0$. \normalsize}
    \State Transfer the items along the path and update the allocation $A^{t}$
        \State $t\gets t+1$
\EndWhile
\State return $A^t$
\end{algorithmic}
\end{algorithm}

The following example demonstrates \Cref{alg:binary-additive}:
\begin{example} \label{app:example: binary additive}
Consider two agents with weights $w_1 = 1$ and $w_2 = 2$, and five items. The valuation functions are:
\[
\begin{bmatrix}
      & o_1 & o_2 & o_3 & o_4 & o_5\\
  i_1 & 1   & 1   & 1   & 1   & 1\\
  i_2 & 1   & 1   & 1   & 1   & 0
\end{bmatrix}
\]
The algorithm is executed as follows:

\begin{enumerate}
  \item For $t=1$, the algorithm compares $\frac{1}{w_1} = \frac{1}{1}$, $\frac{1}{w_2} = \frac{1}{2}$. Consequently, the algorithm searches for a transfer path starting at $i_2$ and ending at $i_{\text{0}}$, and finds the path $(i_2, i_{\text{0}})$. The algorithm transfers the item $o_1$ to agent $i_2$ from $i_{\text{0}}$'s bundle, resulting in $A_1^1 = \emptyset$ and $A_2^1 = \{o_1\}$.
  \item For $t=2$, the algorithm compares $\frac{1}{w_1} = \frac{1}{1}$ and $\frac{v_2(A_2^1) + 1}{w_2} = \frac{2}{2}$. Since those values are equal, the algorithm arbitrarily selects agent $i_2$ and searches
  %\erel{This was not mentioned as a tie-breaking rule}
  %Subsequently, the algorithm searches 
  for a transfer path starting at $i_2$ and ending at $i_{\text{0}}$, and finds the path $(i_2, i_{\text{0}})$. The algorithm transfers the item $o_2$ to agent $i_2$, yielding $A_1^2 = \emptyset$ and $A_2^2 = \{o_1, o_2\}$.
  \item For $t=3$, the algorithm compares $\frac{1}{w_1} = 1$ and $\frac{v_2(A_2^2) + 1}{w_2} = \frac{3}{2}$. As a result, the algorithm searches for a transfer path starting at $i_1$ and ending at $i_{\text{0}}$, and finds the path $(i_1, i_{\text{0}})$. The algorithm transfers the item $o_3$ to agent $i_1$, producing $A_1^3 = \{o_3\}$ and $A_2^3 = \{o_1, o_2\}$.
  \item For $t=4$, the algorithm compares $\frac{v_1(A_1^3) +1}{w_1} = 2$ and $\frac{v_2(A_2^3) + 1}{w_2} = \frac{3}{2}$. Thus, the algorithm searches
   for a transfer path starting at $i_2$ and ending at $i_{\text{0}}$, and finds the path $(i_2, i_{\text{0}})$. The algorithm transfers the item $o_4$ to agent $i_2$, leading $A_1^4 = \{o_3\}$ and $A_2^4 = \{o_1, o_2, o_4\}$.
  \item For $t=5$, the algorithm compares $\frac{v_1(A_1^4) +1}{w_1} = 2$ and $\frac{v_2(A_2^4) + 1}{w_2} = \frac{4}{2} = 2$. Since those values are equal, the algorithm arbitrarily selects agent $i_2$ and searches
  for a transfer path starting at $i_2$ and ending at $i_{\text{0}}$, and finds the path $(i_2, i_1, i_{\text{0}})$. The algorithm transfers the item $o_3$ to agent $i_2$ from $i_1$'s bundle and the item $o_5$ to agent $i_1$ from $i_{\text{0}}$'s bundle, leading $A_1^5 = \{o_5\}$ and $A_2^5 = \{o_1, o_2, o_3, o_4\}$.
  \item Agent 1 envies agent 2 by $\frac{4}{2} - \frac{1}{1} = 1$, while agent 2 envies agent 1 by $0 - \frac{4}{2} < 0$.
  \item In order to mitigate envy, $s_{1} = 1$ and $s_{2} = 0$.
\end{enumerate}
\end{example}

\begin{definition} (\citet{viswanathan2023general})
    An allocation $A$ is said to be \textit{non-redundant} if for all $i\in N$, we have $v_{i}(A_i) = |A_i|$.
\end{definition}
That is, $v_j(A_i) \leq |A_i| = v_i(A_i)$ for every $i,j \in N$.
Similarly to \Cref{thm:wefable--iff-no-cycles1}, we can prove the following:

\begin{lemma} \label{non-redundant is WEF-able}
    With binary valuations, every non-redundant allocation is WEF-able.
\end{lemma}
\begin{proof}
Consider a non-redundant allocation $A$ over $m$ items and $n$ agents with binary additive valuations $v$ and weights $w$. 
Let $C=(i_{1}, ..., i_{r})$ cycle in $G_{A,w}$. 
Then,
\begin{equation}
\begin{split}
    & cost_{A}(C) 
    = \\ & \sum_{k = 1}^{r-1} \left(\frac{v_{i_{k}}(A_{i_{k+1}})}{w_{i_{k+1}}} - \frac{v_{i_{k}}(A_{i_{k}})}{w_{i_{k}}}\right) + \frac{v_{i_{r}}(A_{i_1})}{w_{i_1}} - \frac{v_{i_{r}}(A_{i_r})}{w_{i_r}} \leq \\ 
    & \sum_{k = 1}^{r-1} \left(\frac{v_{i_{k+1}}(A_{i_{k+1}})}{w_{i_{k+1}}} - \frac{v_{i_{k}}(A_{i_{k}})}{w_{i_{k}}}\right) + \frac{v_{i_{1}}(A_{i_1})}{w_{i_1}} - \frac{v_{i_{r}}(A_{i_r})}{w_{i_r}} = \\
    &\sum_{k=1}^{r}\left(\frac{v_{i_{k}}(A_{i_k})}{w_{i_k}} - \frac{v_{i_{k}}(A_{i_k})}{w_{i_k}}\right) = 0,
\end{split}
\end{equation}
where the inequality holds due to non-redundancy. Hence, by condition \ref{condition_b_theroem_1} of \Cref{thm:wefable--iff-no-cycles}, $A$ is WEF-able. 
\end{proof}
%
%The following is an analogue of Observation %\ref{cost_of_path_identical_valuatoins}:
Lemma 3.1 in \cite{viswanathan2023general} shows that the allocation produced by GYS is non-redundant. The same is true for our variant:
\begin{lemma}
\label{app:thm:wefable--iff-no-cycles9} 
At the end of any iteration $t$ of 
\Cref{alg:binary-additive},
the allocation $A^t$ is non-redundant. 
\end{lemma}
\begin{proof}
We prove by induction that at the end of each iteration $t$, $A^t$ remains non-redundant.

For the base case, $A^0$ is an empty allocation and is therefore non-redundant. 
Now, assume that at the end of iteration $t-1$, $A^{t-1}$ is non-redundant. 

If $A^{t} = A^{t-1}$, meaning no agent received a new item, the process is complete. Otherwise, let $i^t$ be the agent who receives new item. Agent $i^t$ obtains an item via the transfer path $P = (i^t = i_1, \ldots, i_k)$. For each $1 \leq j < k$, agent $i_j$ receives the item $o_j$ from the bundle of $i_{j+1}$, given that $v_{i_j}(o_j) = 1$. Agent $i_k$ receives a new item $o_k$ from the bundle of $i_0$, with $v_{i_k}(o_k) = 1$.

Additionally, for each $1 < j \leq k$, item $o_{j-1}$ is removed from agent $i_j$'s bundle where $v_{i_j}(o_{j-1}) = 1$, since $A^{t-1}$ is non-redundant.

For agents not on the transfer path $P$, their bundles remain unchanged. Thus, for each agent $i \in N$, it holds that $v_i(A_i^t) = v_i(A_i^{t-1}) = |A_i^{t-1}| = |A_i^t|$, confirming that $A^t$ is non-redundant.
\end{proof}

    Based on \Cref{app:thm:wefable--iff-no-cycles9} it is established that at the end of every iteration $t\in[T]$, $A^{t}$ is 
    % non-redundant and 
    WEF-able. 
    The remaining task is to establish subsidy bounds. 
    
    We focus on two groups: $R$ and $N \setminus R$. 
    
    The selection rule simplifies limit-setting for $R$ and ensures a subsidy bound of 1 (\Cref{app:Alg4_sub_in_game}). 
    % In the special case of identical binary valuations, all agents remain in the game simultaneously. Therefore, \Cref{app:Alg4_sub_in_game}      ensures an upper bound of $w_i\cdot \frac{1}{w_i} = 1$ on the subsidy of each agent, which implies a total subsidy of at most $n-1$.
%
    However, understanding the dynamics of the second group, $N\setminus R$, presents challenges, as the selection rule is not applicable for them.
    For an agent $i \in N \setminus R$, we prove a subsidy bound of $w_i\cdot \frac{1}{w_j}$,
    for some $j\in R$. In particular, the bound is at most $\frac{w_i}{w_1}$.

 \begin{observation}
\label{obs:non-redundant}
Let $A$ be any non-redundant allocation.
Let $P=(i,\ldots,j)$ be a path in $G_{A,w}$.
Then $cost_A(P) \leq \frac{|A_j|}{w_j} - \frac{|A_{i}|}{w_{i}}$.
\end{observation}

\label{Omitted Details: binary additive}
\begin{lemma}
\label{lem:non-redundant-difference}
    Let $j\in N$ be any agent, if
    $i\in R(t)$, then 
    $
        \frac{|A^t_j|}{w_j}
        -
        \frac{|A^t_i|}{w_i}
        \leq
        \frac{1}{w_i}.
    $
\end{lemma}
\begin{proof}
If $j$ has never been selected to receive an item, then $|A^t_j|=0$ and the lemma is trivial.

Otherwise, let $t'\leq t$ be the latest iteration in which $j$ was selected.
As agents can not be added to $R$ and by the selection rule, 
$\frac{v_j(A_j^{t'-1}) + 1}{w_j} 
\leq 
\frac{v_i(A_i^{t'-1}) + 1}{w_i}$.
Then by non-redundancy, 
$$\frac{|A_j^{t'}|}{w_j} = \frac{v_j(A_j^{t'})}{w_j} = \frac{v_j(A_j^{t'-1}) + 1}{w_j} \leq \frac{v_i(A_i^{t'-1}) + 1}{w_i} = \frac{v_i(A_i^{t'})}{w_i} + \frac{1}{w_i} = \frac{|A_i^{t'}|}{w_i} + \frac{1}{w_i}.$$ 
As $|A_j^{t}|=|A_j^{t'}|$
and $|A_i^{t}| \geq |A_i^{t'}|$,
the lemma follows.
\end{proof}
From \Cref{lem:non-redundant-difference}, we can conclude the following: 
\begin{proposition}\label{app:Alg4_sub_in_game}
If 
 $i\in R(t)$, 
 then 
 $\ell_i(A^t)\leq \frac{1}{w_i}$.
\end{proposition}
\begin{proof}
Assume $P_i= (i, \ldots, j)$ is the path with the highest total cost starting at $i$ in the $G_{A^t,w}$, i.e., $cost_{A^t}(P_i) = \ell_i(A^t)$. 
\Cref{obs:non-redundant} implies $\ell_i(A^t) \leq \frac{|A_j^t|}{w_{j}} - \frac{|A_i^t|}{w_{i}}$. 
As $i\in R(t)$,
\Cref{lem:non-redundant-difference} implies 
$\frac{|A_j^{t}|}{w_j} - \frac{|A_i^t|}{w_{i}} \leq \frac{1}{w_i}$.
\end{proof}

%\Cref{app:Alg4_sub_in_game} ensures an upper bound of $1$ on the subsidy required for agents remaining in the game.
% for each iteration $t\in [T]$ and agent $i\in R(t)$, $\ell_i(A^t) \leq \frac{1}{w_i} \leq \frac{1}{w_1}$.

To prove this upper bound, we need to establish several claims about agents removed from the game. First, we show that an agent removed from the game does not desire any item held by an agent who remains in the game (\Cref{prop: valuation of removed agent}). As a result, these removed agents will not be included in any transfer path (\Cref{prop: removed agent and transfer path}).

Next, we demonstrate that if the cost of a path originating from one of these removed agents at the end of iteration $t$ exceeds the cost at the end of iteration $t' \leq t$ --- the iteration when the agent was removed --- then there exists an edge in this path, $(i_j, i_{j+1})$, such that $v_i(A_j^t) = 0$ (\Cref{at least one 0 valuation}). Based on these claims, we prove that if at the end of iteration $t$, the cost of the maximum-cost path starting from agent removed from the game at $t' < t$ exceeds its cost at $t'$, we can upper-bound it by $\frac{1}{w_1}$. 

\begin{proposition} \label{prop: valuation of removed agent}
    Let $i$ be an agent removed from the game at the start of iteration $t'$.
Then for all $j \in R(t')$, $v_i(A^{t'}_j)=0$.

Moreover, for all $t > t'$
 and all $j\in R(t)$, 
 $v_i(A^{t}_j)=0$.
\end{proposition}
\begin{proof}
    Suppose that $v_i(A_j^{t'}) \neq 0$. This implies there exists some item $o \in A_j^{t'}$ such that $v_i(o) = 1$. We consider two cases:
    \begin{enumerate}
        \item $o \in A_j^{t'-1}$. In this case, at the start of iteration $t'$, there exists a transfer path from $i$ to $j$.
        Moreover, there is a transfer path from $j$ to $i_0$ at the start of iteration $t'$ (otherwise, $j$ would have been removed from the game at $t'$ as well). 
        Concatenating these paths gives a transfer path from $i$ to $i_0$.
        
        \item $o \not \in A_j^{t'-1}$, that is, $j$ received item $o$ during iteration $t'$, from some other agent $j'$ (where $j' = i_0$ is possible).
        At the start of iteration $t'$, there exists a transfer path from $i$ to $j'$.
        Moreover, there is a transfer path from $j'$ to $i_0$, which is used to transfer the newly allocated item. 
        Concatenating these paths gives a transfer path from $i$ to $i_0$.
    \end{enumerate}
    Both cases contradict the assumption that $i$ was removed at $t'$.

To prove the claim for $t>t'$, we use induction over $t$. 
    
We assume the claim holds for iteration $t-1 > t'$ and prove it for iteration $t$. 
Assume, contrary to the claim, that there exists an agent $i$ who was removed at the start of iteration $t'$, and an agent $j \in R(t)$, such that $v_i(A_j^t) \neq 0$. By the induction hypothesis and the fact that an agent can not be added to $R$, we have $v_i(A_j^{t-1}) = 0$. Therefore, during iteration $t$, $j$ must have received a new item $o_j$ that $i$ values at $1$.

    If $o_j$ was part of $i_0$'s bundle at the start of iteration $t'$, then the transfer path starting at $i$, $(i, i_0)$, must have already existed at the start of iteration $t'$.

    Alternatively, if $o_j$ was originally in another agent's bundle at the start of iteration $t'$, say agent $k \in N$,
    then there must have been an iteration between $t'$ and $t$ in which
    $o_j$ has been transferred from $k$ to another agent, while agent $k$ is compensated by some other item $o_k$, that agent $k$ wants.
    
    If $o_k$ was part of $i_0$'s bundle at the start of iteration $t'$, then the transfer path starting at $i$, $(i, k, i_0)$, must have existed at the start of iteration $t'$. Otherwise, $o_k$ was in another agent's bundle at $t'$, and it, too, would be transferred to a different bundle in a later iteration. 
    
    Since the number of items is finite, this process must eventually lead to an item that was in $A_{i_0}^{t'}$, forming a transfer path starting at $i$ at the start of iteration $t'$ --- a contradiction.
    \end{proof}

\begin{proposition}
    \label{prop: removed agent and transfer path}
    Let $i$ be an agent removed from the game at the start of iteration $t'$. Then, for all $t \geq t'$, $i$ will not be included in any transfer path.
\end{proposition}
\begin{proof}
    First, note that $v_i(A_{i_0}^t) = 0$; otherwise, $i$ would not have been removed at the start of iteration $t'$. Next, any agent $j$, who receives an item from $i_0$'s bundle at iteration $t$, must be in $R(t)$. By \Cref{prop: valuation of removed agent}, $v_i(A_j^t) = 0$, which means that $i$ can not receive any item from $j$, who in $R(t)$, as compensation for another item from their own bundle. Thus, any transfer path in iteration $t$ includes only agents in $R(t)$.
\end{proof}

\begin{proposition} \label{at least one 0 valuation}
    Consider an iteration $t$ and an agent $ i \notin R(t)$, who was removed from the game at the start of iteration $t' < t$.  
Let $ P = (i = i_1, \ldots, i_k) $ be a path in the weighted envy graph starting at $ i $.  
If $cost_{A^t}(P) > cost_{A^{t'-1}}(P) $, then there must exist $ j \in \{1, \ldots, k-1\}$ such that $ i_j \notin R $ at $ t $, $ i_{j+1} \in R(t) $ at $t$, and $ v_{i_j}(A_{i_{j+1}}^t) = 0 $.
\end{proposition}

\begin{proof}
Let $t > t'$ be the earliest iteration in which $cost_{A^t}(P) > cost_{A^{t'-1}}(P)$. This implies that there is at least one agent, say $1 \leq j' \leq k-1$, such that agent $i_{j'+1}$ has received a new item that $i_{j'}$ desires. In other words, $i_{j'+1}$ was part of a transfer path at the start of iteration $t$, and by \Cref{prop: removed agent and transfer path}, $i_{j'+1} \in R(t)$ (in particular, $j' \geq 2$).

Since $i_1 \notin R(t)$ and $i_{j'+1} \in R(t)$, there must exist some $1 \leq p < j'+1$ such that $i_p \notin R(t)$ at $t$ and $i_{p+1} \in R(t)$.
By \Cref{prop: valuation of removed agent}, $v_{i_p}(A^t_{i_{p+1}}) = 0$.
\end{proof}

\begin{proposition} \label{subsidy of agent not in the game}
Let $i\notin R(t)$, be an agent who was removed from the game at the start of iteration $t'<t$.
Then, for $A^t$ the resulting allocation from iteration $t$, $\ell_i(A^{t}) \leq \frac{1}{w_1}$.
\end{proposition}
\begin{proof}
Denote by $P_i^{t'-1}, P_i^t$ the highest-cost paths starting from $i$ at iterations $t'-1$ (before agent $i$ removed) and $t$, correspondingly. In particular, 
$$cost_{A^{t'-1}}(P_i^{t}) \leq cost_{A^{t'-1}}(P_i^{t'-1}) = \ell_i(A^{t'-1}).$$ 
Moreover, by \Cref{obs:non-redundant} we have $$cost_{A^{t'-1}}(P_i^{t'-1}) \leq \frac{|A_{i_{k}}^{t'-1}|}{w_{i_{k}}} - \frac{|A_{i}^{t'-1}|}{w_{i}}$$ when $i_k$ is the last agent in $P_i^{t'-1}$. Combined with \Cref{lem:non-redundant-difference}, this gives $\ell_i(A^{t'-1}) \leq \frac{1}{w_i}$.
Therefore, if
$cost_{A^{t}}(P_i^t) = \ell_i(A^t) \leq \ell_i(A^{t'-1})$, then we are done. 

Assume now that $\ell_i(A^t) > \ell_i(A^{t'-1})$. 
Denote the path $P_i^t$ by $(i = i_1,\ldots, i_k)$.

From \Cref{at least one 0 valuation}, there exists $j \in \{1,\ldots,k-1\}$ such that $i_j\notin R(t)$, $i_{j+1}\in R(t)$ and $v_{i_{j}}(A_{i_{j+1}}^t) = 0$.
Then, by \Cref{obs:non-redundant}: 
   \begin{align}
       & \ell_i(A^t) = 
        cost_{A^t}(i,...,i_j) + cost_{A^t}(i_{j}, i_{j+1}) + cost_{A^t}(i_{j+1},...,i_k) \leq \nonumber \\
       &\leq \left(\frac{|A_{i_j}^t|}{w_{i_j}} - \frac{|A_i^t|}{w_i}\right) + \left(0 - \frac{|A_{i_j}^t|}{w_{i_j}}\right) + \ell_{i_{j+1}}(A^t) \leq  \nonumber\\
       \label{cost of maximum path}
       & \ell_{i_{j+1}}(A^t).
   \end{align}
   Since $i_{j+1} \in R(t)$, it follows from \eqref{cost of maximum path} and \Cref{app:Alg4_sub_in_game} that \[\ell_i(A^t) \leq \ell_{i_{j+1}}(A^t) \leq \frac{1}{w_{i_{j+1}}} \leq 
   \frac{1}{w_1}.\]
\end{proof}
    
\begin{theorem}\label{app:theorem_21}
    For additive binary valuations, \Cref{alg:binary-additive} computes a WEF-able allocation where the subsidy to each agent $i\in N$ is at most $\frac{w_i}{w_1}$ in polynomial-time.
    Moreover, the total subsidy is bounded by $\frac{W}{w_1} - 1$.
\end{theorem}
\begin{proof}
    Together \Cref{app:Alg4_sub_in_game} and \Cref{subsidy of agent not in the game} establish that for every $i\in N$ and $t \in [T]$, $\ell_i(A^t)\leq \frac{1}{w_1}$. Along with \Cref{app:thm:wefable--iff-no-cycles9}, \Cref{alg:binary-additive} computes a WEF-able allocation $A^{T}$ where the required subsidy per agent $i \in N$ is at most $\frac{w_i}{w_1}$. 

As there is at least one agent who requires no subsidy (see \Cref{max_path_subsidy}), the total required subsidy is at most $\frac{W-w_1}{w_1} = \frac{W}{w_1} - 1$.
 
We complete the proof of \Cref{app:theorem_21} by demonstrating that \Cref{alg:binary-additive} runs in polynomial-time. 
We represent the valuations using a binary matrix $X$ where $v_i(o_j) = 1 \Longleftrightarrow X(i,j) = 1$. Hence, the allocation of items to the bundle of $i_0$ at line 1 can be accomplished in $O(mn)$ time.

    At each iteration $t$ of the while loop, either $A_{i_0}$ or $R$ reduced by 1, ensuring that the loop runs at most $m+n$ times.

    Let $T_v$ represent the complexity of computing the value of a bundle of items, and $T_\phi$ denote the complexity of computing the gain function. Both are polynomial in $m$.

According to \citet{viswanathan2023general}, finding a transfer path starting from agent $i \in N$ (or determining that no such path exists) takes $O(T_v \log m)$. Removing agents at the start of each iteration incurs a complexity of $O(n T_v \log m)$. Furthermore, as stated in \citet{viswanathan2023general}, identifying $i^t$ requires $O(nT_v)$. Updating the allocation based on the transfer path, according to the same source, takes $O(m)$.

Thus, each iteration has a total complexity of $$O(nT_v \log m + nT_v + T_v \log m + m) = O\left( nT_v \log m + m\right).$$

    In conclusion, \Cref{alg:binary-additive} runs in $O\left(\left(m+n\right)\left(nT_v \log m + m\right)\right)$, which is polynomial in both $m$ and $n$.
    \end{proof}
    In Appendix \ref{alg:binary-additive tightness}, we present a tighter bound that is closer to the lower bound, along with a detailed discussion on its tightness.
    
Notice that since the output allocation from \Cref{alg:binary-additive} is non-redundant, $A$ maximizes the social welfare. 
Moreover, as shown in \Cref{app:example: binary additive}, $A$ might not be \emph{WEF(1,0)} (No matter which item is removed from $i_2$'s bundle, $i_1$ still envies). 
% Hence we can conclude that an allocation maximizing the social welfare is not necessarily \emph{WEF(1,0)}.
However, it is \emph{WEF(0,1)}. 
\begin{proposition} \label{app:binary additive wef01}
    For additive binary valuations, \Cref{alg:binary-additive} computes a $WEF(0,1)$ allocation.
\end{proposition}
\begin{proof}
We prove by induction that at the end of each iteration $t\in [T]$
    $A^t$ satisfies $WEF(0,1)$. This means that for every $i,j \in N$, there exists a set of items $B \subseteq A_j^t$ of size at most $1$ such that $\frac{v_i(A_i^t) + v_i(B)}{w_i} \geq \frac{v_i(A_j^t)}{w_j}$.

The claim is straightforward for the first iteration. We assume the claim holds for the $(t-1)$-th iteration and prove it for the $t$-th iteration. Note that $v_i(A_i^{t-1}) \leq v_i(A_i^{t})$.
\begin{enumerate}
    \item $A_j^t = A_j^{t-1}$ and $A_i^t = A_i^{t-1}$: the claim holds due to the induction step.
    \item $v_i(A_j^t) = v_i(A_j^{t-1})$: This is the case where $j$ was not included in a transfer path, or was included but was not the first agent in the path, and exchanged an item for a new one, both having the same value for $i$. 
    By the induction assumption, there exists some singleton $B^{t-1} \subseteq A_j^{t-1}$
    such that $\frac{v_i(A_i^{t-1}) + v_i(B^{t-1})}{w_i} \geq \frac{v_i(A_j^{t-1})}{w_j} = \frac{v_i(A_j^t)}{w_j}$.
    There exists some singleton $B^{t} \subseteq A_j^{t}$, with $v_i(B^t) = v_i(B^{t-1})$.
    Hence, 
    $$\frac{v_i(A_i^{t}) + v_i(B^t)}{w_i} \geq \frac{v_i(A_i^{t-1}) + v_i(B^{t-1})}{w_i} \geq \frac{v_i(A_j^{t-1})}{w_j} = \frac{v_i(A_j^t)}{w_j}.$$
    \item  $v_i(A_j^t) < v_i(A_j^{t-1})$: This is the case where $j$ was included in a transfer path but exchanged an item $i$ values for an item $i$ does not value. Then $$\frac{v_i(A_i^{t}) + v_i(B)}{w_i} \geq \frac{v_i(A_i^{t-1}) + v_i(B)}{w_i} \geq \frac{v_i(A_j^{t-1})}{w_j} > \frac{v_i(A_j^t)}{w_j}$$
        for a set $B \subseteq A_j^t$ of size at most $1$.
    \item $v_i(A_j^t) > v_i(A_j^{t-1})$: there are two subcases:
    \begin{enumerate}
        \item If $j$ is the first agent in the transfer path and received a new item $o$ such that $v_i(o) = 1$, then 
        $\frac{v_i(A_i^{t-1}) + 1}{w_i} \geq \frac{v_j(A_j^{t-1}) + 1}{w_j}$ due to the selection rule,
        and 
        $\frac{v_j(A_j^{t-1}) + 1}{w_j}
        \geq \frac{v_i(A_j^{t-1}) + 1}{w_j}$ due to non-redundancy.
        
        We can conclude that $$\frac{v_i(A_i^{t}) + 1}{w_i} \geq \frac{v_i(A_i^{t-1}) + 1}{w_i} \geq \frac{v_i(A_j^{t-1}) + 1}{w_j} = \frac{v_i(A_j^{t})}{w_j}.$$ The claim holds for $B = \{o\}\subseteq A_j^t$.
        \item If $j$ was not the first agent in the path, but exchanged an item that $i$ does not value for an item $o$ that $i$ values, $v_i(o) = 1$.
        Let $t' < t$ represent the most recent iteration in which agent $j$ was selected and received a new item. Note that $\frac{v_i(A_i^{t}) + 1}{w_i} \geq \frac{v_i(A_i^{t'-1}) + 1}{w_i} \geq \frac{v_j(A_j^{t'-1}) + 1}{w_j}$ due to the selection rule.
        
        Assume to the contrary that $\frac{v_i(A_i^t) + 1}{w_i} < \frac{v_i(A_j^t)}{w_j}$. Then,
        \begin{align*}
            & \frac{v_i(A_i^t) + 1}{w_i} < \frac{v_i(A_j^t)}{w_j} \leq \frac{v_j(A_j^{t})}{w_j} = \frac{v_j(A_j^{t'})}{w_j} =  \\ 
            &\frac{v_j(A_j^{t'-1}) + 1}{w_j} \leq \frac{v_i(A_i^{t'-1}) + 1}{w_i} \leq \frac{v_i(A_i^{t}) + 1}{w_i},
        \end{align*}
        a contradiction.
        Hence, $\frac{v_i(A_i^t) + 1}{w_i} \geq \frac{v_i(A_j^t)}{w_j}$ and the claim holds for $B = \{o\}\subseteq A_j^t$.
    \end{enumerate}
\end{enumerate}
% \erel{I do not see how this follows from the induction step, if the object $o_t$ is different at each step?} \noga{fixed}
\end{proof}
\section{Conclusions And Future Work}
We studied the problem of subsidy minimization required for achieving weighted envy-freeness among agents with varying entitlements when allocated indivisible items.

Previous work in the unweighted setting of subsidies relied on basic characterizations of EF that fail in the weighted settings. This poses interesting challenges in the new setting. 

We show that an allocation is WEF-able only if its weighted envy graph does not contain positive cost cycles. Unlike the unweighted scenario (\citet{halpern2019fair}), %MWUSW is not a condition for WEF.
a given allocation that maximizes the utilitarian welfare across all reassignments of its bundles to agents is not a condition for WEF-able. 
% \erel{Does the citation show that it is EF?} \noga{fixed}
The gap between the weighted setting and the unweighted setting in EF characterization raises the open question of other conditions under which weighted EF allocations may exist.

We've shown polynomial-time algorithms to compute WEF-able allocations for general, identical, and binary additive valuations in the weighted setting. The proved subsidy bounds are $\left(W-w_1\right) V$, $(n-1) V$, and $\frac{W}{w_1} - 1$, respectively. 

While our $(n-1) V$ bound is tight, some gap remains for $\left(W-w_1\right) V$ and  $\frac{W}{w_1} - 1$ for the weighted setting. However for identical and binary additive valuations, our bounds align with those of the unweighted setting. Like the previous literature, we focus on additive valuations. This highlights the need for further investigation into refining subsidy bounds for the general additive case and extending results to non-additive valuations in both weighted and unweighted contexts.  

\iffalse

We extended our analysis of weighted allocations to include \emph{weighted-leximin} and \emph{MWNW} allocations. For binary additive valuations, we demonstrated that the \emph{weighted-leximin} allocation is not only \emph{MWNW} and \emph{WEF-able}, but also requires a total subsidy of at least $W-1$. In contrast, under identical additive valuations, the \emph{weighted-leximin} allocation does not necessarily align with the \emph{MWNW} allocation, though it remains \emph{WEF-able} and requires a total subsidy of at least $\left(W-1\right)  V$. The situation becomes more complex under general additive valuations, where no clear connection exists between \emph{weighted-leximin}, \emph{MWNW}, and \emph{WEF-able} concepts, presenting an intriguing challenge for further investigation.

We extended our fairness study to generic truthful VCG, demonstrating that, unlike in the unweighted case, the W-VCG mechanism is not WEF. However, we introduced the utility equalization mechanism, a combination of allocation and payment vectors, that ensures both WEF and weighted truthful (WIC) for agents with binary additive valuations.  We leave for future research the exploration of other mechanisms that are WIC and WEF in the case of agents with general additive valuations.
\fi
\nocite{*}

\newpage

\bibliography{mybibfile}
\appendix

\section*{APPENDIX}
\section{Tightness of the Subsidy Bounds}
\subsection{Subsidy Bound of \Cref{alg:general-additive}} \label{alg:general-additive tightness}
     As \Cref{theorem: sub general additive} implies, \Cref{alg:general-additive} computes a WEF-able allocation with a total subsidy of at most $(W-w_1)V$. However, this
     bound is not tight. To understand why, consider the case of $2$ items, each valued at $V$ by agent $i\in \{1, \ldots, n-1\}$, who has an entitlement of $w_i \geq 2$, and $V - \epsilon$ by all other agents. Our algorithm will allocate all the items to agent $i$, resulting in a subsidy of $\frac{w_j}{w_i} 2 (V-\epsilon)$ by each other agent $j\neq i \in N$, leading to a total subsidy of $(W-w_i)\frac{2(V-\epsilon)}{w_i}$, for arbitrarily small $\epsilon>0$.
     
     In general, a WEF-able allocation can achieve a lower subsidy by allocating one item to another agent with higher index $j > i$, i.e, $w_j \geq w_i$. For instance, if one item is allocated to such agent $j$, agent $i$ envies agent $j$ by $\frac{V}{w_j} - \frac{V}{w_i} \leq 0$, and agent $j$ envies agent $i$ by $\frac{V - \epsilon}{w_i} - \frac{V - \epsilon}{w_j} < \frac{2(V - \epsilon)}{w_i}$. If $\frac{V - \epsilon}{w_i} - \frac{V - \epsilon}{w_j} \leq 0$, then no subsidy is required. Otherwise, the subsidy required by agent $j$ is $\left( \frac{V - \epsilon}{w_i} - \frac{V - \epsilon}{w_j} \right) w_j < \frac{w_j}{w_i} \cdot 2(V - \epsilon)$. The subsidy required by each other agent $k \neq i,j$ is significantly lower than $w_k \cdot \frac{w_j}{w_j} \cdot \frac{2(V - \epsilon)}{w_i} = \frac{w_k}{w_i} \cdot 2(V - \epsilon)$. Therefore, the required total subsidy is significantly lower than $(W - w_i) \frac{2(V - \epsilon)}{w_i}$.

     In both cases, the resulting total subsidy bound is better than the bound obtained by allocating all items to agent $i$.

     \subsection{Subsidy Bound of \Cref{alg:binary-additive}} \label{alg:binary-additive tightness}
As \Cref{app:theorem_21} implies, \Cref{alg:binary-additive} computes a WEF-able allocation with a total subsidy of at most $\frac{W}{w_1} - 1$. However, with more careful analysis, we can prove a tighter bound.

There are two cases to consider:
\begin{enumerate}
    \item \textbf{Agent 1 with the minimum entitlement receives a positive subsidy.} Together, \Cref{app:Alg4_sub_in_game} and \Cref{subsidy of agent not in the game} imply that $s_i \leq \frac{w_i}{w_1}$ for each agent $i \in N$. Since agent 1 does receive a positive subsidy, and by \Cref{max_path_subsidy}, there exists at least one agent who requires no subsidy, the total required subsidy is bounded by $\frac{W - w_2}{w_1}$.
    \item \textbf{Agent 1 with the minimum entitlement receives no subsidy.} We can modify \Cref{subsidy of agent not in the game} in the following way: for each agent $i \notin R(t)$, where $t \in [T]$, $\ell_i(A^t) \leq \frac{1}{w_2}$. By the proof of \Cref{subsidy of agent not in the game}, $\ell_i(A^t) \leq \ell_{i_{j+1}}(A^t)$. If $i_{j+1} = i_1$, then $\ell_i(A^t) \leq \ell_{i_{j+1}}(A^t) \leq 0$ (because agent 1 requires no subsidy). Otherwise, $\ell_i(A^t) \leq \ell_{i_{j+1}}(A^t) \leq \frac{1}{w_{i_{j+1}}} \leq \frac{1}{w_2}$.
    Overall, the subsidy required by each agent is bounded by $\frac{w_i}{w_2}$, and by \Cref{max_path_subsidy}, there exists at least one agent who requires no subsidy. Therefore, the total required subsidy is bounded by $\frac{W - w_1}{w_2}$.
\end{enumerate}
To sum up, the total required subsidy is at most $\max\Big\{\frac{W - w_1}{w_2}, \frac{W - w_2}{w_1}\Big\}$.
\end{document}